\documentclass[a4paper,UKenglish,cleveref,autoref,thm-restate]{lipics-v2021}
\usepackage{tikz}
\usetikzlibrary{intersections, calc, arrows, automata, positioning}

\RequirePackage{tikz}
\usetikzlibrary{arrows,snakes}

\usepackage{subcaption}
\usepackage{graphicx}

\usepackage{pgf, verbatim}
\usetikzlibrary{arrows,automata}

\usepackage[linesnumbered,ruled]{algorithm2e}
\usepackage{amsmath}
\allowdisplaybreaks

\usepackage{floatrow}
\newfloatcommand{capbtabbox}{table}[][\FBwidth]
\usepackage{tabularx}

\usepackage{xcolor,colortbl}
\usepackage{longtable}

\usepackage[font=small,labelfont=bf]{caption}

\usepackage{mathtools}
\usepackage{amssymb}
\usepackage{stmaryrd}

\usepackage{amsthm}

\newcommand{\commenteq}[1]{\hspace{2em} [\mbox{#1}]}

\newcommand{\suchthat}{:}

\usepackage[capitalise]{cleveref}

\mathcode`<="4268
\mathcode`>="5269
\mathchardef\gr="213E
\mathchardef\ls="213C

\newcommand{\Dist}{\mathrm{Distr}}
\newcommand{\SubDist}{\mathrm{SubDistr}}
\newcommand{\support}{\mathrm{support}}
\newcommand{\lifting}[1]{\mathord{#1\!\!\uparrow}}
\newcommand{\nat}{\mathbb{N}}
\newcommand{\integer}{\mathbb{Z}}

\newcommand{\local}{\mathit{local}}

\newcommand{\nxt}{\mathit{next}}
\newcommand{\M}{\mathcal{M}}

\newcommand{\Hyp}{\mathcal{M}'}
\newcommand{\errorParam}{\epsilon_2}
\newcommand{\tauHyp}{\tau_{\epsilon}}
\newcommand{\R}{\mathcal{R}}
\newcommand{\Q}{\mathcal{Q}}

\definecolor{OliveGreen}{rgb}{0,0.6,0}

 \newcommand{\modify}[1]{{\color{black}#1}}

 \title{Approximate Bisimulation Minimisation} 

\author{Stefan Kiefer}{Department of Computer Science, University of Oxford, UK}{stekie@cs.ox.ac.uk}{https://orcid.org/0000-0003-4173-6877}{Supported by a Royal Society University Fellowship.}

\author{Qiyi Tang}{Department of Computer Science, University of Oxford, UK}{qiyi.tang@cs.ox.ac.uk}{https://orcid.org/0000-0002-9265-3011}{}

\authorrunning{S.\ Kiefer and Q.\ Tang}

\Copyright{Stefan Kiefer and Qiyi Tang} 



 \category{} 

\relatedversion{The paper is accepted to FSTTCS 2021.}

\acknowledgements{We thank the anonymous reviewers of this paper for their constructive feedback.}

\nolinenumbers 

\EventEditors{Miko{\l}aj Boja\'{n}czyk and Chandra Chekuri}
\EventNoEds{2}
\EventLongTitle{41st IARCS Annual Conference on Foundations of Software Technology and Theoretical Computer Science (FSTTCS 2021)}
\EventShortTitle{FSTTCS 2021}
\EventAcronym{FSTTCS}
\EventYear{2021}
\EventDate{December 15--17, 2021}
\EventLocation{Virtual Conference}
\EventLogo{}
\SeriesVolume{213}
\ArticleNo{28}

\begin{document}\sloppy
	
\maketitle
	
\begin{abstract}
We propose polynomial-time algorithms to minimise labelled Markov chains whose transition probabilities are not known exactly, have been perturbed, or can only be obtained by sampling.
Our algorithms are based on a new notion of an approximate bisimulation quotient, obtained by lumping together states that are exactly bisimilar in a slightly perturbed system.
We present experiments that show that our algorithms are able to recover the structure of the bisimulation quotient of the unperturbed system.
\end{abstract}
	
\section{Introduction}

For the algorithmic analysis and verification of system models, computing the bisimulation quotient is a natural preprocessing step: it can make the system much smaller while preserving most properties of interest.
This applies equally to probabilistic systems: probabilistic model checkers, e.g., Storm~\cite{Storm20}, speed up the verification process by ``lumping'' together states that are equivalent with respect to probabilistic bisimulation.
While this is a safe approach, it may not be effective when the probabilities in the system are not known precisely.
\begin{figure}[ht]
	\centering
	\tikzstyle{BoxStyle} = [draw, circle, fill=black, scale=0.4,minimum width = 1pt, minimum height = 1pt]
	
	\begin{tikzpicture}[scale=.6,>=latex',shorten >=1pt,node distance=3cm,on grid,auto]
	
	
	\node[state] (s) at (-1,0) {$s_1$};
	\node[state, fill=green] (s2) at (4,0) {$s_2$};
	
	\node[state] (t) at (9,0) {$t_1$};
	\node[state, fill=green] (t2) at (14,0) {$t_2$};
	
	
	\path[->] (s) edge [out=30,in=150] node [midway, above] {$\frac{1}{2}$} (s2);
	\path[->] (s2) edge [out=210,in=-30] node [midway, below] {$\frac{1}{2}$} (s);
	
	\path[->] (t) edge [out=30,in=150] node [midway, above] {$\frac{1}{2}-\epsilon$} (t2);
	\path[->] (t2) edge  [out=210,in=-30]  node [midway, below] {$\frac{1}{2}-\epsilon$} (t);
	
	
	\path[->] (s) edge [loop above]node [midway, above] {$\frac{1}{2}$} (s);
	\path[->] (s2) edge [loop above] node [midway, above] {$\frac{1}{2}$} (s2);
	
	\path[->] (t) edge [loop above] node [midway, above] {$\frac{1}{2}+\epsilon$} (t);
	\path[->] (t2) edge [loop above] node [midway, above] {$\frac{1}{2}+\epsilon$} (t2);
	
	
	\end{tikzpicture}
\caption{Two intuitively similar labelled Markov chains.}
\label{fig:intro3}
\end{figure}
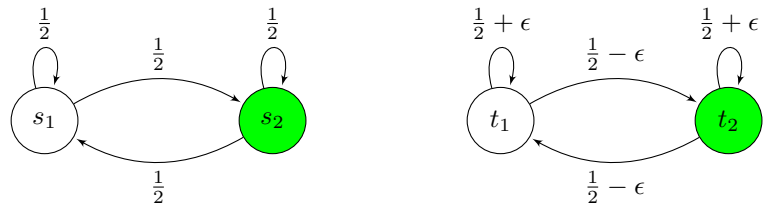
For example, in the labelled Markov chain shown in \cref{fig:intro3} the states $s_1, t_1$ are intuitively ``similar'', but they are not probabilistically bisimilar even though they carry the same label (here indicated with the colour white) and they both lead, with similar probabilities, to states $s_2, t_2$, which are again intuitively ``similar'' but not probabilistically bisimilar.

For analysis and verification of a probabilistic system, tackling state space explosion is key.
The more symmetry a system has (e.g., due to variables that do not influence the observable behaviour of the model), the greater are the benefits of computing a \emph{quotient} system with respect to probabilistic bisimulation.
However, when the transition probabilities in the Markov chain are perturbed or not known exactly, bisimulation minimisation may fail to capture the symmetries in the system and thus fail to achieve the objective of making the probabilistic system amenable to algorithmic analysis.
The motivation of this paper is to counter this problem.

A principled approach towards this goal may consider notions of \emph{distance} between probabilistic systems or between states in a probabilistic system.
A \emph{probabilistic bisimilarity pseudometric}, based on the Kantorovich metric, was introduced in \cite{DGJP1999,DesharnaisGJP04}, which assigns to each pair of states $s,t$ a number in the interval $[0,1]$ measuring a distance between $s,t$: distance~$0$ means probabilistic bisimilar, and distance~$1$ means, in a sense, maximally non-bisimilar.
This pseudometric can be computed in polynomial time \cite{ChenBW12}, and, quoting \cite{ChenBW12}, has ``been studied in the context of systems biology, games, planning and security, among others''.
The corresponding distances can be intuitively large though: the pseudometric yields a distance less than~$1$ only if the two states can reach, with the same label sequence, two states that are \emph{exactly} bisimilar; see \cite[Section~4]{TangVanBreugel2018}.
As a consequence, for any $\epsilon \gr 0$, the states $s_1, t_1$ in \cref{fig:intro3} have distance~$1$ in the probabilistic bisimilarity pseudometric of \cite{DesharnaisGJP04} (in the undiscounted version).
From a slightly different point of view, a small perturbation of the transition probabilities in the model can change the distance between two states from $0$ to~$1$.

Another pseudometric, \emph{$\epsilon$-bisimulation~$\mathord{\sim_{\epsilon}}$}, was defined in \cite{DesharnaisLavoletteTracol2008}, which avoids this issue.
It has natural characterisations in terms of games and can be computed in polynomial time using network-flow based algorithms \cite{DesharnaisLavoletteTracol2008}.
The runtime of the algorithm from \cite{DesharnaisLavoletteTracol2008} is $O(|S|^7)$, where $S$ is the set of states, thus not practical for large systems.
A more fundamental problem lies in the fact that $\epsilon$-bisimulation is not an equivalence: $s \sim_\epsilon t \sim_\epsilon u$ implies $s \sim_{2 \epsilon} u$ (by the triangle inequality) but not necessarily $s \sim_{\epsilon} u$.
Therefore, efficient minimisation algorithms via quotienting (such as partition refinement for exact probabilistic bisimilarity) are not available for $\epsilon$-bisimulation.

In this paper we develop algorithms that, given a labelled Markov chain~$\M$ with possibly imprecise transition probabilities and a \emph{slightly perturbed} version, say $\Hyp$, of~$\M$, compute a compressed version, $\Q$, of~$\Hyp$.
By slightly perturbed we mean that for each state the successor distributions in $\Hyp$ and~$\M$ have small (say, less than~$\epsilon$) $L_1$-distance.
We hope that $\Q$ is not much bigger than the exact quotient of~$\M$, and we design polynomial-time algorithms that fulfill this hope in practice, but we do not insist on computing the smallest possible~$\Q$.
Indeed, we show that, given an LMC, $\epsilon \gr 0$ and a positive integer $k$, it is {\sf NP}-complete to decide whether there exists a perturbation of at most~$\epsilon$ such that the (exact) bisimulation quotient of the perturbed system is of size $k$.
See, e.g., \cite{BacciBLM18} for an exact but non-polynomial approach, where the target number~$n$ of states is fixed, and a Markov chain with at most~$n$ states is sought that has minimal distance (with respect to the probabilistic bisimilarity pseudometric of \cite{DesharnaisGJP04}) to the given model.

It is not hard to prove (\cref{proposition:approximate-quotient-implies-approximate-bisimulation}) that if an LMC can be made exactly bisimilar to another LMC by a perturbation of at most~$\epsilon$, then these two LMCs are also $\frac{\epsilon}{2}$-bisimilar in the sense of \cite{DesharnaisLavoletteTracol2008}.
If, in turn, two states are $\epsilon$-bisimilar, one can show (see \cite[Theorem~4]{BianA17}) that any linear-time property that depends only on the first $k$ labels has similar probabilities in the two states, where similar means the difference in probabilities is at most $1 - (1-\epsilon)^k$. 
Combining these two results, we obtain a continuity property that says if two LMCs can be made bisimilar by a small perturbation, then any $k$-bounded linear-time property has similar probabilities in the two LMCs.
In other words, our approximate minimisation approximately preserves the probability of bounded linear-time properties.

We apply our approximate minimisation algorithms in a setting of active learning.
Here we assume we do not have access to the transition probabilities of the model; rather, for each state we only \emph{sample} the successor distribution.
Sampling gives us an approximation of the real Markov chain, and our approximate minimisation algorithms apply naturally.
This allows us to lump states that are exactly bisimilar in the real model, but only approximately bisimilar in the sampled model.
We give examples where in this way we \emph{recover} the structure (not the precise transition probabilities) of the quotient of the exact model, knowing only the sampled model.

The rest of the paper is organised as follows. In \cref{section:preliminaries} we introduce $\epsilon$-quotient, a new notion of approximate bisimulation quotient. In \cref{section:approximate-quotient-properties} given an LMC, $\epsilon$ and $k \gr 0$, we show it is ${\sf NP}$-complete to decide whether there exists an $\epsilon$-quotient of size $k$. In \cref{section:minimisation-algorithms} we present our approximate minimisation algorithms.
We put them in a context of active learning in \cref{section:active-LMC-learning}.
In \cref{section:experiments} we evaluate these algorithms on slightly perturbed versions of a number of LMCs taken from the probabilistic model checker PRISM \cite{KNP11}. We conclude in \cref{section:conclusion}. 
\section{Preliminaries}
\label{section:preliminaries}

We write $\nat$ for the set of nonnegative integers and $\integer^{+}$ for the set of positive integers. We write $\mathbb{R}$ for the set of real numbers. Let $S$ be a finite set. We denote by $\Dist(S)$ the set of probability distributions on~$S$. By default we view vectors, i.e., elements of $\mathbb{R}^{S}$, as row vectors. For a vector $\mu \in\mathbb{R}^{S}$ we write \modify{$\|\mu\|_1 := \sum_{s \in S} |\mu(s)|$ for its $L_1$-norm. A vector $\mu \in [0, 1]^{S}$ is a distribution (resp. subdistribution) over $S$ if $\|\mu\|_1 = 1$ (resp.  $0 \ls \|\mu\|_1 \le 1$).} For a (sub)distribution $\mu$ we write $\support(\mu) = \{s \in S \suchthat \mu(s) \gr 0 \}$ for its support. 


A partition of the states $S$ is a set $X$ consisting of pairwise disjoint subsets $E$ of $S$ with $\bigcup_{E \in X} = S$. For an equivalence relation $\R \subseteq S \times S$, $S /_{\R}$ denotes its quotient set and $[s]_\R$ denotes the $\R$-equivalence class of $s \in S$.

A \emph{labelled Markov chain} (LMC) is a quadruple $<S, L, \tau, \ell>$ consisting of a nonempty finite set $S$ of states, a nonempty finite set $L$ of labels, a transition function $\tau : S \to \Dist(S)$, and a labelling function $\ell: S \to L$.

We denote by $\tau(s)(t)$ the transition probability from $s$ to $t$. Similarly, we denote by $\tau(s)(E) = \sum_{t \in E} \tau(s)(t)$ the transition probability from $s$ to $E \subseteq S$.
For the remainder of the paper,  we fix an LMC $\M = <S,  L, \tau, \ell>$. Let $|\M|$ denote the number of states, $|S|$.

The direct sum $\M_1 \oplus \M_2$ of two LMCs $\M_1 = <S_1, L_1, \tau_1, \ell_1>$ and $\M_2 = <S_2, L_2, \tau_2, \ell_2>$ is the LMC formed by combining the state spaces of $\M_1$ and $\M_2$.

An equivalence relation $\R \subseteq S \times S$ is a \emph{probabilistic bisimulation} if for all $(s, t) \in \R$, $\ell(s) = \ell(t)$ and
$\tau(s)(E) = \tau(t)(E)$ for each $\R$-equivalence class $E$. \emph{Probabilistic bisimilarity}, denoted by $\mathord{\sim_{\M}}$ (or $\mathord{\sim}$ when $\M$ is clear), is the largest probabilistic bisimulation.

Any probabilistic bisimulation $\R$ on $\M$ induces a quotient LMC denoted by $\M/_{\R} = <S/_{\R}, L, \tau/_{\R}, \ell/_{\R}>$ where the transition function $\tau/_{\R}([s]_{\R})([t]_{\R}) = \tau(s)([t]_{\R})$ and the labelling function $\ell/_{\R}([s]_\R) = \ell(s)$.

We define the notion of an \emph{approximate quotient}.  Let $\epsilon \ge 0$. An LMC $\Q$ is an \emph{$\epsilon$-quotient} of $\M$ if and only if there is transition function $\tau': S \to \Dist(S)$ such that for all $s \in S$ we have $\|\tau'(s) - \tau(s)\|_1 \le \epsilon$ and $\Q$ is the (exact) bisimulation quotient of the LMC $\M' = <S, L, \tau',\ell>$, that is, $\Q = \M'/_{\sim_{\M'}} $. Since the choice of $\tau'$ is not unique, there might be multiple $\epsilon$-quotients of $\M$. We are interested in the problem of obtaining an $\epsilon$-quotient of $\M$ with small state space. We retrieve the notion of (exact) quotient when $\epsilon = 0$. For $s$ from $\M$, denote the state in $\Q$ which corresponds to $s$ by $[s]^{\epsilon}_{\Q}$ (or $[s]^{\epsilon}$ when $\Q$ is clear).

The set $\Omega(\mu, \nu)$ of \emph{couplings} of $\mu,\nu \in \Dist(S)$ is defined as $\Omega(\mu, \nu)=\left \{ \, \omega \in \Dist(S \times S) \suchthat \sum_{t \in S} \omega(s, t) = \mu(s) \wedge \sum_{s \in S} \omega(s, t) = \nu(t) \, \right \}$. Note that a coupling $\omega \in \Omega$ is a joint probability distribution with marginals $\mu$ and $\nu$ (see, e.g., \cite[page 260-262]{billingsley1995}).



	The \emph{$\epsilon$-lifting} of a relation $\R\subseteq S \times S$ proposed by Tracol et al.~\cite{TracolDesharnaisZhioua2011} is the relation $\lifting{\R}_{\epsilon} \subseteq \Dist(S) \times \Dist(S)$ defined by $(\mu, \nu) \in \lifting{\R}_{\epsilon}$ if there exists $\omega \in \Omega(\mu, \nu)$ such that $\sum_{ (u,v) \in \R} \omega(u, v) \ge 1 - \epsilon$.

The \emph{$\epsilon$-bisimulation} ($\sim_{\epsilon}$) by Desharnais et al.~\cite{DesharnaisLavoletteTracol2008} is a relation $\R \subseteq S \times S$ in which for all $(s, t) \in \R$ we have  $\ell(s) = \ell(t)$ and $(\tau(s), \tau(t)) \in \lifting{\R}_{\epsilon}$. The $\epsilon$-bisimulation is reflexive and symmetric, but in general not transitive; hence, it is not an equivalence relation.


\section{Properties of Approximate Quotients}
\label{section:approximate-quotient-properties}

\begin{figure}[h]
	\centering
	
	\begin{tikzpicture}[xscale=.6,>=latex',shorten >=1pt,node distance=3cm,on grid,auto]
	
	\node[label] (M) at (0,0) {$\M$};
	\node[label] (MQ) at (8,0) {$\M /_{\sim_{\M}}$};
	\node[label] (Me) at (0,-1) {$\Hyp$};
	\node[label] (MeQ) at (8,-1) {$\Q$};
	
	\path[->] (M) edge node [midway, above] {quotient} (MQ);
	\path[->] (M) edge node [midway, left] {perturbation} (Me);
	\path[->] (Me) edge node [midway, above] {approximate quotient} (MeQ);		
	\path[dashed] (MeQ) edge node [midway, right] {$\Q$ is not much bigger than $\M /_{\sim_{\M}}$} (MQ);		
	\end{tikzpicture}
	\caption{Problem setup.}\label{fig:problem-setup}
\end{figure}

Recall from the introduction that we are given an LMC $\Hyp$, which is a slightly perturbed version of an (unknown) LMC $\M$. By slightly perturbed we mean that for each state the successor distributions in $\Hyp$ and~$\M$ have small (say, less than~$\epsilon$) $L_1$-distance. For example, with sampling we can obtain with high probability a perturbed system that has small distance with $\M$. Assume there are many symmetries, that is, lots of probabilistic bisimilar states in $\M$. The state space of $\M$ can then be compressed a lot by (exact) quotienting. Since the transition probabilities are perturbed in $\Hyp$, the states that are probabilistic bisimilar in $\M$ might become inequivalent in $\Hyp$; as a result, the (exact) bisimulation quotient of $\Hyp$ is much larger than that of $\M$. Given a small compression parameter $\errorParam \gr 0$, we aim to compute an approximate quotient $\Q$, an $\epsilon'$-quotient of $\Hyp$ that satisfies two conditions: \modify{(1) $\epsilon'$ should be small, so that little precision is
sacrificed; and (2) the state space of the quotient should be small, to speed up
verification algorithms. Our contribution consists of approximate quotienting algorithms with (a) theoretical guarantees on goal (1) in \cref{theorem:bounding-global-distance} and
\cref{corollary:bounding-quotient-error}, applying to both algorithms: $\epsilon'$ is bounded (and can be controlled) by a compression parameter $\epsilon_2$ and
the number of iterations $i$; (b) empirical results on goal (2): the experiments show that our algorithms produce small quotients.} 

We first show that it is hard to find an $\errorParam$-quotient of $\Hyp$ with minimum number of states: $\Q^{*}= \arg\min\{|\Q| : \Q \text{ is an } \errorParam\text{-quotient of }\Hyp\}.$ If there are several $\errorParam$-quotients of $\Hyp$ of minimum size, $\Q^{*}$ can be taken to be any one of them. Unfortunately, this problem is unlikely to have an efficient (polynomial-time) solution, as we will see from the next theorem that the decision version of this problem is $\sf NP$-complete.

\begin{figure}[h]
	\begin{minipage}{0.45\linewidth}
		\centering
		
		\begin{tikzpicture}[xscale=.6,>=latex',shorten >=1pt,node distance=3cm,on grid,auto]
		
		\node[state] (us) at (3,4) {$s$};
		\node[state] (u1) at (0.5,2.5) {$s_{1}$};
		\node[label] at (3,2.5) {$\cdots$};
		\node[state] (un) at (5.5,2.5) {$s_{n}$};

		\node[state] (qb) at (0.5,1) {$s_a$};
		\node[state, fill=green] (qc) at (5.5,1) {$s_b$};
		
		\path[->] (us) edge node  [midway,above] {$\frac{p_1}{T}$} (u1);
		\path[->] (us) edge node  [midway,above] {$\frac{p_n}{T}$} (un);

		\path[->] (qc) edge  [loop right] node [right] {$1$} (qc);
		\path[->] (qb) edge [loop left] node [left]{$1$} (qb);
		
		\path[->] (u1) edge node  [midway,left] {$\frac{1}{2}$} (qb);
		\path[->] (u1) edge node  [very near start, below,xshift=0.1cm] {$\frac{1}{2}$} (qc);
		\path[->] (un) edge node [very near start, below,xshift=-0.1cm] {$\frac{1}{2}$} (qb);
		\path[->] (un) edge node [midway,right] {$\frac{1}{2}$} (qc);
		
		\end{tikzpicture}
		
	\end{minipage}
	\hspace{0.5cm}
	\begin{minipage}{0.45\linewidth}
		\centering
		
		\begin{tikzpicture}[xscale=.6,>=latex',shorten >=1pt,node distance=3cm,on grid,auto]
		
		\node[state] (t) at (2,4) {$t$};
		
		\node[state] (t1) at (-0.5,2.5) {$t_1$};
		\node[state] (t2) at (4.5,2.5) {$t_2$};
		
		\node[state] (qb) at (-0.5,1) {$t_a$};
		\node[state, fill=green] (qc) at (4.5,1) {$t_b$};
		
		\path[->] (t) edge node [midway, above] {$\frac{N}{T}$} (t1);
		\path[->] (t) edge node [midway, above, xshift=0.25cm] {$1-\frac{N}{T}$} (t2);
		
		\path[-] (t1) edge node [midway, right] {} (qb);
		\path[->] (t1) edge node [near start, left] {$\frac{1}{2} - \epsilon$} (qb);
		\path[->] (t1) edge node [near start, above,xshift=0.05cm,yshift=0.1cm] {$\frac{1}{2} + \epsilon$} (qc);
		
		\path[-] (t2) edge node [midway, right] {} (qc);		
		\path[->] (t2) edge node [near start, above, yshift=0.1cm] {$\frac{1}{2}+ \epsilon$} (qb);
		\path[->] (t2) edge node [near start, right,xshift=0.15cm] {$\frac{1}{2} - \epsilon$} (qc);
		
		\path[->] (qc) edge [loop right] node [midway, right] {$1$} (qc);
		\path[->] (qb) edge [loop left]   node [midway, left]{$1$} (qb);
		
		\end{tikzpicture}
	\end{minipage}
	\caption{The LMC in the reduction for {\sf NP}-hardness. All states have the same label $a$ except $s_b$ and $t_b$ which have label $b$.
	}\label{fig:reductionfromSubsset}
\end{figure}
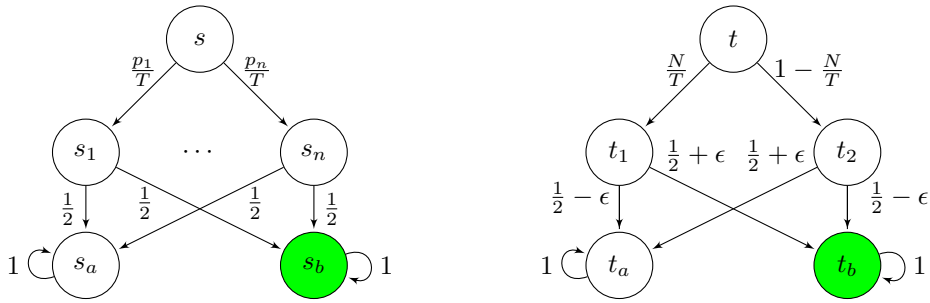

Given an LMC $\Hyp$, a compression parameter $\errorParam \gr 0$ and a constant $k \in \integer^{+}$, it is {\sf NP}-complete to decide whether there exists an $\errorParam$-quotient of $\Hyp$ of size $k$. The hardness result is by reduction from the Subset Sum problem. Given a set $P =\{p_1, \ldots, p_n\}$ and $N \in \nat$, Subset Sum asks whether there exists a set $Q \subseteq P$ such that $\sum_{p_i \in Q} p_i = N$. Given an instance of Subset Sum $<P, N>$ where $P =\{p_1, \ldots, p_n\}$ and $N \in \nat$, we construct an LMC; see \cref{fig:reductionfromSubsset}. Let $T = \sum_{p_i \in P} p_i$, $\epsilon = \frac{1}{2T}$ and $k = 5$. \modify{In the LMC, state $s$ transitions to state $s_i$ with probability $p_i / T$ for all $1 \le  i \le n$. Each state $s_i$ transitions to $s_a$ and $s_b$ with equal probabilities. State $t$ transitions to $t_1$ and $t_2$ with probability $N / T$ and $1 - N / T$, respectively. State $t_1$ (resp. $t_2$) transitions to $t_a$ (resp. $t_b$) and $t_b$ (resp. $t_a$) with probability $\frac{1}{2} - \epsilon$ and $\frac{1}{2} + \epsilon$, respectively. All the remaining states transition to the successor state with probability one. States $s_b$ and $t_b$ have label $b$ and all other states have label $a$.} We can show that $<P, N> \in {\mbox{Subset Sum}} \iff$ there exists an $\frac{1}{2T}$-quotient of $\Hyp$ of size $5$.

\begin{restatable}{theorem}{theoremMinimumApproximateQuotientNPComplete}\label{theorem: minimum-approximate-quotient-NP-complete}
	Given an LMC $\Hyp$, $\errorParam \in (0, 1]$ and $k \in \integer^{+}$. The problem whether there exists an $\errorParam$-quotient of $\Hyp$ of size $k$ is $\sf NP$-complete. It is $\sf NP$-hard even for (fixed) $k=5$.
\end{restatable}

Due to the $\sf NP$-hardness result, we hope to develop practical algorithms to compute approximate quotients of $\Hyp$ that are small but not necessarily of minimum size. To do that, an intuitive idea is to merge ``similar'' states. As we have discussed in the introduction, merging states with small probabilistic bisimilarity distances might be insufficient. Consider the LMC shown in \cref{fig:intro3}. Assume $\epsilon \gr 0$. The states $s_1$ and $t_1$ ($s_2$ and $t_2$) have probabilistic bisimilarity distance one. Thus, to merge $s_1$, $t_1$ or $s_2$, $t_2$, one needs to merge states with probabilistic bisimilarity distance one. Alternatively, we explore the relation of approximate quotient and $\epsilon$-bisimulation. It is not hard to prove the following proposition:

\begin{restatable}{proposition}{propositionApproximateGlobalRelationSubset}\label{proposition:approximate-quotient-implies-approximate-bisimulation}
	Let $\Q$ be an $\errorParam$-quotient of $\Hyp$. Then in the LMC $\Hyp \oplus \Q$, we have $s \sim_{\frac{\errorParam}{2}} [s]^{\errorParam}_{\Q}$ for all $s$ from $\Hyp$.
\end{restatable}

\begin{figure}[t]
	\centering
	\begin{tikzpicture}
	\tikzstyle{BoxStyle} = [draw, circle, fill=black, scale=0.05,minimum width = 0.001pt, minimum height = 0.001pt]
	\node[state] (s) at (3,4) {$s_1$};
	\node[state] (s1) at (2,2.5) {$s_{2}$};
	\node[state, fill=green] (x) at (4,2.5) {$x$};
	\path[->] (s) edge [loop left] node  [midway,left] {$\frac{1}{2}$} (s);
	\path[->] (s) edge node  [midway,left,xshift=-0.1cm,yshift=0.1cm] {$\frac{1}{4}$} (s1);
	\path[->] (s) edge node  [midway,right,xshift=0.1cm,yshift=0.1cm] {$\frac{1}{4}$} (x);
	\path[->] (x) edge [loop right] node [midway, right] {$1$} (x);
	\path[->] (s1) edge node  [midway,below] {$\frac{1}{4}-2\epsilon$} (x);
	\path[->] (s1) edge [loop left] node [midway, left] {$\frac{3}{4} + 2\epsilon$} (s1);
	
	\node[state] (t) at (8,4) {$s_3$};
	\node[state, fill=green] (y) at (8,2.5) {$x$};
	\path[->] (t) edge [loop left] node  [midway,left] {$\frac{3}{4}+\epsilon$} (t);
	\path[->] (t) edge node  [midway,right,xshift=0.1cm] {$\frac{1}{4}-\epsilon$} (y);
	\path[->] (y) edge [loop right] node [midway, right] {$1$} (y);
	
	\end{tikzpicture}
	\caption{An LMC in which $s_1 \sim_{\epsilon} s_3 \sim_{\epsilon} s_2$.}
	\label{fig:exampleGlobal}
\end{figure}
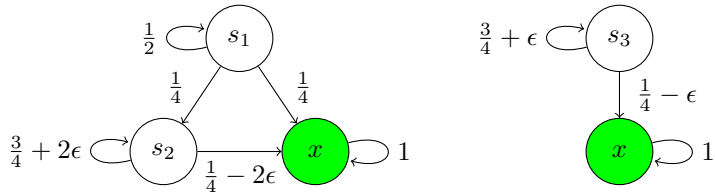
\begin{figure}[t]
	\centering
	\begin{tikzpicture}
	\tikzstyle{BoxStyle} = [draw, circle, fill=black, scale=0.05,minimum width = 0.001pt, minimum height = 0.001pt]
	\node[state] (s) at (3,4) {$s_{13}$};
	\node[state] (s1) at (2,2.5) {$s_{2}$};
	\node[state, fill=green] (x) at (4,2.5) {$x$};
	\path[->] (s) edge [loop left] node  [midway,left] {$\frac{5}{8}+\frac{\epsilon}{2}$} (s);
	\path[->] (s) edge node  [midway,left,xshift=-0.1cm,yshift=0.1cm] {$\frac{1}{8}$} (s1);
	\path[->] (s) edge node  [midway,right,xshift=0.1cm,yshift=0.1cm] {$\frac{1}{4}-\frac{\epsilon}{2}$} (x);
	\path[->] (x) edge [loop right] node [midway, right] {$1$} (x);
	\path[->] (s1) edge node  [midway,below] {$\frac{1}{4}-2\epsilon$} (x);
	\path[->] (s1) edge [loop left] node [midway, above] {$\frac{3}{4} + 2\epsilon$} (s1);
	\node at (3, 1.5) {(a) $\epsilon'$-quotient};
	
	\node[state] (s) at (8,4) {$s_1$};
	\node[state] (s1) at (7,2.5) {$s_{23}$};
	\node[state, fill=green] (x) at (9,2.5) {$x$};
	\path[->] (s) edge [loop left] node  [midway,left] {$\frac{1}{2}$} (s);
	\path[->] (s) edge node  [midway,left,xshift=-0.1cm,yshift=0.1cm] {$\frac{1}{4}$} (s1);
	\path[->] (s) edge node  [midway,right,xshift=0.1cm,yshift=0.1cm] {$\frac{1}{4}$} (x);
	\path[->] (x) edge [loop right] node [midway, right] {$1$} (x);
	\path[->] (s1) edge node  [midway,below] {$\frac{1}{4}-\frac{3\epsilon}{2}$} (x);
	\path[->] (s1) edge [loop left] node [midway, above,yshift=0.1cm] {$\frac{3}{4} + \frac{3\epsilon}{2}$} (s1);	
	\node at (8, 1.5) {(b) $\epsilon$-quotient};
	
	\node[state] (t) at (12,4) {$s_{123}$};
	\node[state, fill=green] (y) at (12,2.5) {$x$};
	\path[->] (t) edge [loop left] node  [midway,left] {$\frac{3}{4}+\epsilon$} (t);
	\path[->] (t) edge node  [midway,right,xshift=0.1cm] {$\frac{1}{4}-\epsilon$} (y);
	\path[->] (y) edge [loop right] node [midway, right] {$1$} (y);
	\node at (12, 1.5) {(c) $2\epsilon$-quotient};	
	\end{tikzpicture}
	\caption{(a) An $\epsilon'$-quotient obtained by merging $s_1$ and $s_3$ where $\epsilon'$ is at least $\frac{1}{4}+\epsilon$; (b) An $\epsilon$-quotient obtained by merging $s_2$ and $s_3$; (b) A $2\epsilon$-quotient obtained by merging $s_1$, $s_2$ and $s_3$.}
	\label{fig:exampleMerge}
\end{figure}
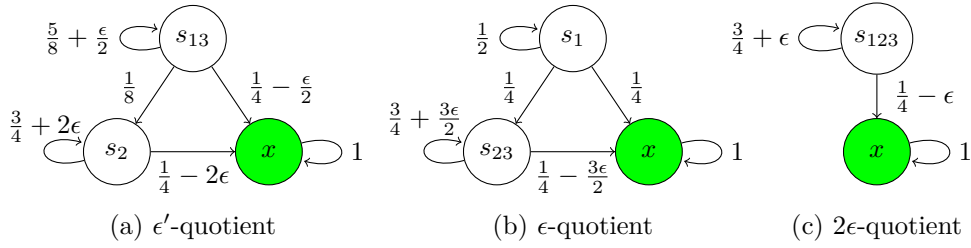

\cref{proposition:approximate-quotient-implies-approximate-bisimulation} suggests that $\epsilon_2$-quotients and $\epsilon_2$-bisimulation are related.  The runtime of the algorithm to compute the $\epsilon_2$-bisimulation in \cite{DesharnaisLavoletteTracol2008} is $O(|S|^7)$ which makes it not practical for large systems. Furthermore, the algorithms based on merging states that are $\epsilon_2$-bisimilar may produce an $\epsilon'$-quotient where $\epsilon'$ is large, violating the first condition of a satisfying approximate quotient. Assume the positive number $\epsilon$ is much smaller than $\frac{1}{8}$. Let us choose the compression parameter $\epsilon_2$ to be the same as $\epsilon$. We compute the $\epsilon$-bisimulation of the LMC shown in \cref{fig:exampleGlobal} and get $s_1 \sim_{\epsilon} s_3 \sim_{\epsilon} s_2$.  Since $\epsilon$-bisimulation is not an equivalence relation, $s_1  \sim_{\epsilon} s_2$ does not necessarily follow. Indeed, in this LMC, we have $s_1  \sim_{2\epsilon} s_2$ but not $s_1  \sim_{\epsilon} s_2$. If $s_2$ and $s_3$, related by $\sim_{\epsilon}$, are chosen to be merged, the resulting LMC in \cref{fig:exampleMerge}(b) is an $\epsilon$-quotient. However, if $s_1$ and $s_3$ are (unfortunately) chosen to be merged, the resulting LMC, shown in \cref{fig:exampleMerge}(a), is an $\epsilon'$-quotient where $\epsilon'$ cannot be smaller than $\frac{1}{4}+\epsilon$. This $\epsilon'$, much bigger than $\epsilon$ under the assumption that $\epsilon$ is much smaller than $\frac{1}{8}$, makes the resulting LMC undesirable. This example shows that arbitrarily merging states that are $\epsilon$-bisimilar may not work. The LMC in \cref{fig:exampleMerge}(c) is obtained by merging $s_1$, $s_2$ and $s_3$, the states that are related by the transitive closure of $\sim_{\epsilon}$. We show in the appendix that for any $n \in \integer^{+}$ there exists an LMC $\M(n)$ such that merging all states in $\M(n)$  that are related by the transitive closure of $\sim_{\epsilon}$ results in an $\epsilon'$-quotient where $\epsilon'$ is at least $n\epsilon$.


\cref{lemma:additivity-property}, the additivity lemma, asserts an additivity property of approximate quotients. In \cref{section:minimisation-algorithms}, this lemma will be applied as the two minimisation algorithms successively compute a sequence of approximate quotients.


\begin{restatable}{lemma}{lemmaAdditivityProperty}\label{lemma:additivity-property}
		Consider three LMCs $\M_1$, $\M_2$ and $\M_3$. Let $\epsilon_1 \ge 0$ and $\M_2$ be an $\epsilon_1$-quotient of $\M_1$. Let $\epsilon_2 \ge 0$ and $\M_3$ be an $\epsilon_2$-quotient of $\M_2$. Then $\M_3$ is an $(\epsilon_1+\epsilon_2)$-quotient of $\M_1$. 
\end{restatable}

\section{Approximate Minimisation Algorithms}\label{section:minimisation-algorithms}

\begin{figure}[h]
	\centering
	
	\begin{tikzpicture}[xscale=.6,>=latex',shorten >=1pt,node distance=3cm,on grid,auto]
	
	\node[label] (Me) at (0,0) {$\Hyp$};
	\node[label] (Q0) at (4,0) {$\Q_{0}$};
	\node[label] (Q1) at (9,0) {$\Q_{1}$};
	\node[label] (Qd) at (14,0) {$\cdots$};
	\node[label] (Qi) at (19,0) {$\Q_{i}$};
	
	\path[->] (Me) edge node [midway, above] {exact} (Q0);
	\path[->] (Me) edge node [midway, below] {quotient} (Q0);	
	\path[->] (Q0) edge node [midway, above] {approximate} (Q1);
	\path[->] (Q0) edge node [midway, below] {quotient} (Q1);	
	\path[->] (Q1) edge node [midway, above] {approximate} (Qd);	
	\path[->] (Q1) edge node [midway, below] {quotient} (Qd);				
	\path[->] (Qd) edge node [midway, above] {approximate} (Qi);		
	\path[->] (Qd) edge node [midway, below] {quotient} (Qi);		
	\end{tikzpicture}
	\caption{Overview of the minimisation algorithms. \cref{lemma:additivity-property} applies to $\Hyp$, $\Q_{0}$, $\Q_{1}, \cdots, \Q_{i}$.}\label{fig:algorithms}
\end{figure}

In this section, we present two practical minimisation algorithms that compute approximate quotients of $\Hyp$. Given an LMC $\Hyp = <S, L, \tauHyp, \ell>$ with perturbed transition probabilities and a small compression parameter $\epsilon_2$. Both algorithms start by computing $\Q_0$, the exact quotient of $\Hyp$.  They proceed in iterations and compute a sequence of approximate quotients where the approximate quotient ($\Q_{i}$) computed at the end of the $i$th iteration is an $\epsilon_2$-quotient of the quotient ($\Q_{i-1}$) given at the beginning of that iteration. Using the additivity lemma, we can show that the (approximate) quotient $\Q_i$ after the $i$th iteration is an $i\epsilon_2$-quotient of $\Hyp$. See \cref{fig:algorithms} for an overview of this approach. Each iteration computes a partition of the state space, lumps the states that are together in the partition and concludes with taking the exact quotient.


\subsection{Local Bisimilarity Distance}\label{subsection:local-bisimilarity-distances}

We define the notion of \emph{local bisimilarity distance}, denoted by $d_{\local}^{\M}$ (or $d_{\local}$ when $\M$ is clear). Intuitively, two states $s$ and $t$ are at small local bisimilarity distance if they are probabilistic bisimilar in an LMC which is slightly perturbed only at the successor distributions of $s$ and~$t$. We provide a polynomial-time algorithm to compute the local bisimilarity distance. Given an LMC $\Hyp = <S, L, \tauHyp, \ell>$ (with perturbed transition probabilities) and a small compression parameter $\epsilon_2$, we propose an iterative minimisation algorithm to compute approximate quotients of $\Hyp$ by merging state pairs with small local bisimilarity distances. In each iteration of the algorithm, we select the state pair with the same label and the minimum local bisimilarity distance if such distance is at most $\epsilon_2$. We compute a partition in which this state pair are together and lump together the states that are together in the partition. The algorithm terminates when no pairs can be lumped, that is, all state pairs have local bisimilarity distances greater than $\epsilon_2$. 


 \paragraph*{Computing Local Bisimilarity Distances}\label{subsubsection:computeLocalBisimilarityDistances}
Given two different states $s, t \in S$ with the same label. We want to compute a new transition function $\tauHyp'$ by only changing the successor distributions of $s$ and $t$ ($\tauHyp(s)$ and $\tauHyp(t)$, respectively) such that $\{s, t\}$ belongs to an $\R$-induced partition where $\R$ is a probabilistic bisimulation of the LMC $\Hyp' = <S, L, \tauHyp', \ell>$. Let ${\rm T}$ be the set of the all transition functions that satisfy this condition, more precisely, we define ${\rm T} = \{\tauHyp' \suchthat \tauHyp'(x) = \tauHyp(x) \; \forall x \not\in \{s, t\}  \land \{s, t\} \in S /_{\R} \text{ where } \R \text{ is a probabilistic bisimulation of the LMC } \Hyp' = <S, L, \tauHyp', \ell>\}$. The local bisimilarity distance is defined as $d_{\local}^{\Hyp}(s, t) = \textstyle\inf_{\tau' \in {\rm T}} \max\{\|\tau'(s) - \tauHyp(s)\|_1, \|\tau'(t) - \tauHyp(t)\|_1 \}$. It is not immediately clear how to compute it. 

By the definition of ${\rm T}$, the probabilistic bisimulation $\R$ is the same for any LMC $<S, L, \tauHyp', \ell>$ with $\tauHyp' \in {\rm T}$. Let us define the partition $X = S /_\R$ where $\R$ is the common probabilistic bisimulation. The local bisimilarity distance can be computed by using $X$:

\begin{restatable}{proposition}{propositionAdjustTransitionFunction}\label{proposition:adjust-transition-function-s-t}
We have $d_{\local}^{\Hyp}(s, t) =	 \frac{1}{2} \|(\tau_{\epsilon}(s)(E))_{E \in X} -(\tau_{\epsilon}(t)(E))_{E \in X}\|_1$.
\end{restatable}

It turns out that $X$ can simply be computed by \cref{alg:local-distance-partition}. As this algorithm is basically taking the (exact) quotient of the LMC constructed on line~1, it runs in polynomial time. It follows from \cref{proposition:adjust-transition-function-s-t} that the local bisimilarity distance can be computed in polynomial time.


\begin{algorithm}[h]
	\DontPrintSemicolon
	\KwIn{An LMC $\Hyp = <S, L, \tauHyp, \ell>$, a state pair $(s, t) \in S \times S$}
	\KwOut{A partition $X$ over $S$ containing $\{s, t\}$}
	Construct a new LMC $\Hyp'$ from $\Hyp$ by introducing a new label, labelling both $s$ and $t$ with the new label and making both $s$ and $t$ absorbing\footnotemark\;
	$X := S /_{\sim_{\Hyp'}}$\;
	\caption{Compute Partition for Local Bisimilarity Distances}
	\label{alg:local-distance-partition}
\end{algorithm}
\footnotetext{\modify{An absorbing state is a state that, once entered, cannot be left; that is, a state with self-loop.}}

\begin{example}\label{example:intro3}
 Assume $\epsilon \ls \frac{1}{2}$. Consider the LMC shown in \cref{fig:intro3}. Let $\tauHyp$ denote its transition function. To compute the local bisimilarity distance of $s_1$ and $t_1$, we first compute the partition containing $\{s_1, t_1\}$: $X = \big\{ \{ s_1, t_1\}, \{ s_2\},\{ t_2\} \big\}$. We have $(\tauHyp(s_1)(E))_{E \in X} = (\frac{1}{2}, \frac{1}{2}, 0)$ and $(\tauHyp(t_1)(E))_{E \in X} = (\frac{1}{2}+\epsilon, 0, \frac{1}{2}-\epsilon)$. By \cref{proposition:adjust-transition-function-s-t}, the local bisimilarity distance is $d_{\local}(s_1, t_1) = \frac{1}{2}\|(\tauHyp(s_1)(E))_{E \in X} -(\tauHyp(t_1)(E))_{E \in X}\|_1 = \frac{1}{2}$. Similarly, we have~$d_{\local}(s_2, t_2) =~\frac{1}{2}$.
\end{example}

\begin{algorithm}[h]
	\DontPrintSemicolon
	\KwIn{An LMC $\Hyp = <S, L, \tauHyp, \ell>$, a compression parameter $\epsilon_2$}
	\KwOut{An LMC $\Q_{i}$}
	$i := 0$\\
	$\Q_{i} := \Hyp/_{\sim_{\Hyp}}$ and $\Q_{i} = <S^{\Q_{i}}, L, \tau^{\Q_{i}}, \ell^{\Q_{i}}>$\;
	\While{$\exists u, v \in S^{\Q_{i}} \text{ such that }  u \not= v \text{ and }  \ell^{\Q_{i}}(u) = \ell^{\Q_{i}}(v) \text{ and }  d_{\local}^{\Q_{i}} (u, v) \le \epsilon_2$ }{
		$(s, t) = \arg\min \{d_{\local}^{\Q_{i}} (u, v) \suchthat (u ,v) \in S^{\Q_{i}} \times S^{\Q_{i}} \land u \not= v \land   \ell^{\Q_{i}}(u) = \ell^{\Q_{i}}(v)\}$\\
		Compute $X_i$ by running \cref{alg:local-distance-partition} with input $Q_i$ and $(s, t)$\;
		Construct an LMC $\M_{i+1}:= <X_i, L, \tau^{\M_{i+1}}, \ell^{\M_{i+1}}>$ from $\Q_{i}$ where
		$
		\tau^{\M_{i+1}}(E) := \left \{
		\begin{array}{l}
		 (\tau^{\Q_{i}}(u)(E'))_{E' \in X_i}  \mbox{ for any $u \in E$ if $E \in X_i$  and $E \not= \{s, t\}$}\\
		\frac{(\tau^{\Q_{i}}(s)(E'))_{E' \in X_i} + (\tau^{\Q_{i}}(t)(E'))_{E' \in X_i}}{2} \;\; \mbox{if $E = \{s, t\}$}\\
		\end{array}
		\right .
		$ and $\ell^{\M_{i+1}} (E) := \ell^{\Q_i}(u)$ for $E \in X_i$ and any $u \in E$\;
		
		$\Q_{i+1} := \M_{i+1} /_{\sim_{\M_{i+1}}}$\;
		$i := i+1$\;
	}
	\caption{LMC Minimisation Using Local Bisimilarity Distances}
	\label{alg:local-distance-merge-algorithm}
\end{algorithm}

\paragraph*{Minimisation Algorithm Using Local Bisimilarity Distances}
\label{subsubsection:minimisation-algorithm-local-bisimilarity-distances}
\cref{alg:local-distance-merge-algorithm} shows the minimisation algorithm using local bisimilarity distances. The input is an LMC $\Hyp$ and a compression parameter $\epsilon_2$. We start by initializing an index $i$ to $0$ and building the quotient LMC $\Q_{0} = \Hyp /_{\sim_{\Hyp}}$. If there are no states in $\Q_{i}$ with local bisimilarity distance less than $\epsilon_2$, the algorithm terminates. Otherwise, it steps into the $i$'th iteration of the loop and computes the local bisimilarity distances for all pairs of states in $\Q_{i}$ with the same label. It selects the state pair $(s, t)$ which has the smallest local bisimilarity distance on line~4. It then computes the new approximate quotient by merging states $s$ and $t$ on line~$5$-$7$. This computation is in three steps where the first step is to compute the partition $X_i$ (line~5) by running \cref{alg:local-distance-partition} with input $\Q_i$ and the state pair $(s, t)$. The second step is to construct a new LMC $\M_{i+1}$ by setting $X_i$ as its state space (line~6). The final step is to compute a new approximate quotient $\Q_{i+1}$ by taking the exact quotient of the LMC $\M_{i+1}$ obtained from the previous step.
We increment $i$ at the end of the iteration and continue with another iteration if there are states in $\Q_{i+1}$ with local bisimilarity distance at most $\epsilon_2$. Since there are finitely many states and it is polynomial time to compute the local bisimilarity distances, the algorithm always terminates and runs in polynomial time.
\subsection{Minimisation by Approximate Partition Refinement}\label{subsection:approximate-partition-refinement}
Consider the LMC in \cref{fig:intro3}. Assume $\epsilon \ls \frac{1}{2}$ and $\epsilon_2 \ls \frac{1}{2}$. The minimisation algorithm using local bisimilarity distance (Algorithm~\ref{alg:local-distance-merge-algorithm}) cannot merge states $s_1, t_1$ (or $s_2, t_2$) as $d_{\local}(s_1, t_1) = d_{\local}(s_2, t_2)  =\frac{1}{2} \gr \epsilon_2$ as shown by \cref{example:intro3}.

We introduce an approximate partition refinement, a polynomial algorithm similar to the exact partition refinement, which can fix this problem. In the exact partition refinement algorithm, the states will only remain in the same set in an iteration if they have the same label and their probability distributions over the previous partition are the same. Similarly, we design the approximate partition refinement such that states only remain in the same set in an iteration if they have the same label and the $L_1$-distance between the probability distributions over the previous partition is small, say, at most $\epsilon_2$. Given an LMC $\Hyp = <S, L, \tauHyp, \ell>$ with perturbed transition probabilities, the minimisation algorithm using the approximate partition refinement also proceeds in iterations. In each iteration, the approximate partition refinement computes a partition $X$ and then the states which are together in $X$ are  lumped to form a new LMC. To make sure the new LMC is a quotient, we take the (exact) quotient of this LMC as our new approximate quotient. The algorithm continues when there are states that could be lumped, and it terminates when all sets in the partition computed by the approximate partition refinement are singletons, that is, no states can be lumped. 

\modify{
	\begin{example}\label{example:approximate-partition-intro}	
	Consider again the LMC in \cref{fig:intro3}. Assume $\epsilon \ls \frac{1}{2}$ and the compression parameter $\epsilon_2 \ge 2\epsilon$.  We run the above-mentioned minimisation algorithm using the approximate partition refinement. It will only run for one iteration of approximate partition refinement, as we will see in the following. \Cref{fig:example-intro-approximate-partition-refinement}(a) shows the partitions of this iteration. At the beginning of the approximate partition refinement, we have partition $X_0$ as all states are in the same set. The states are then split by the labels and we get partition $X_1$. There is no further split since the $L_1$-distance between the probability distributions over $X_1$ from $s_1$ and $t_1$ (resp. $s_2$ and $t_2$) is $2\epsilon$ which is bounded by the compression parameter $\epsilon_2$, that is, $\|(\tau(s_1)(E))_{E \in X_1}  - (\tau(t_1)(E))_{E \in X_1}\|_1 = \|(\tau(s_2)(E))_{E \in X_1}  - (\tau(t_2)(E))_{E \in X_1}\|_1 = 2\epsilon \le \epsilon_2$. The states together in $X_1$ are then lumped to form the new LMC shown in \cref{fig:example-intro-approximate-partition-refinement}(b). The algorithm terminates as no states in the new LMC can be lumped.
	\end{example}
}

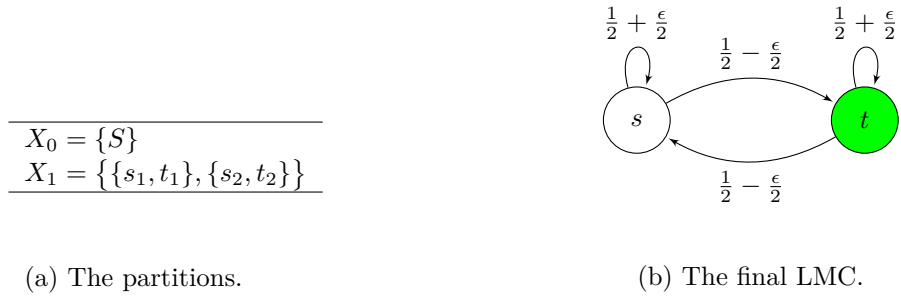
\begin{figure}
\begin{minipage}{0.45\textwidth}
	\centering	\vspace{0.2cm}
	\begin{tabular}{l}
		\\\\\\
		\hline
		$X_0 = \{S\}$ \\
		$X_1 = \big\{  \{s_1,t_1\} , \{s_2,t_2\} \big\}$\\
		\hline
		\\\\
		(a) The partitions.
	\end{tabular}
\end{minipage}
\hfill
\begin{minipage}{0.45\textwidth}
	\centering
\tikzstyle{BoxStyle} = [draw, circle, fill=black, scale=0.4,minimum width = 1pt, minimum height = 1pt]

\begin{tikzpicture}[scale=.6,>=latex',shorten >=1pt,node distance=3cm,on grid,auto]


\node[state] (s) at (-1,0) {$s$};
\node[state, fill=green] (s2) at (4,0) {$t$};


\path[->] (s) edge [out=30,in=150] node [midway, above] {$\frac{1}{2}-\frac{\epsilon}{2}$} (s2);
\path[->] (s2) edge [out=210,in=-30] node [midway, below] {$\frac{1}{2}-\frac{\epsilon}{2}$} (s);

	\path[->] (s) edge [loop above]node [midway, above] {$\frac{1}{2} +\frac{\epsilon}{2}$} (s);
\path[->] (s2) edge [loop above] node [midway, above] {$\frac{1}{2} +\frac{\epsilon}{2} $} (s2);

\node at (1.5,-3.5) {(b) The final LMC.};
\end{tikzpicture}
\end{minipage}
\caption{Example of running the minimisation algorithm using approximate partition refinement (Algorithm~\ref{alg:polynomial-optimistic-partition-refinement}) on the LMC in \cref{fig:intro3}.} \label{fig:example-intro-approximate-partition-refinement}
\end{figure}

\paragraph*{Approximate Partition Refinement}
\begin{algorithm}[h]
	\DontPrintSemicolon
	\KwIn{An LMC $\Hyp = <S, L, \tauHyp, \ell>$, a compression parameter $\errorParam$}
	\KwOut{A partition $X$ over $S$}	
	$i := 0; X_0 := \{S\}$\;
	\Repeat{$X_i = X_{i-1}$}{
		$i := i+1$; $X_i := \emptyset$\;
		\ForEach{$E \in X_{i-1}$}{
			$X_E := \emptyset$\;
			\For{$s \in E$}{
				$ESet := \{E'\in X_{E} \suchthat \text{for all } t\in E' \text{ we have } \ell(s)=\ell(t) \text{ and}\linebreak \|(\tauHyp(s)(E))_{E \in X_{i-1}} - (\tauHyp(t)(E))_{E \in X_{i-1}}\|_1 \le \epsilon_2 \}$\;
				\lIf{$ESet = \emptyset$}{
					$E': = \{s\}$
				}\Else{
					$E' := \arg\min\limits_{E' \in ESet}\big\{ \frac{\sum_{t \in  E'} \|(\tauHyp(s)(E))_{E \in X_{i-1}} - (\tauHyp(t)(E))_{E \in X_{i-1}}\|_1}{|E'|} \big\}$\;
					remove $E'$ from $X_E$; $E' := E' \cup \{ s \}$\;
				}
				add $E'$ to $X_E$ \;
			}
			$X_i := X_i \cup X_E$\;
		}
	}
	\caption{\mbox{Approximate Partition Refinement}}
	\label{alg:polynomial-optimistic-partition-refinement}
\end{algorithm}

Given a compression parameter $\epsilon_{2}$, the approximate partition refinement is shown in Algorithm~\ref{alg:polynomial-optimistic-partition-refinement}. At the beginning, an index $i$ is initialized to zero and we have $X_0 = \{S\}$, that is, all states are in the same set. In a refinement step, we increment $i$ and split each set $E \in X_{i-1}$ into one or more sets. We iterate though all $E \in X_{i-1}$ and for each $E$ we construct a set $X_E$, a partition of $E$. Starting with $X_E = \emptyset$, we iterate over all $s \in E$ (line~6). After each iteration, the current $s \in E$ appears in one set in $X_E$: either as a singleton or as an additional state in an already existing set in $X_E$. We give more details on this loop (lines~6-14) below.
After having partitioned~$E$ into~$X_E$, we add all sets in~$X_E$ to the new partition~$X_i$.
The way we split the sets ensures that for any two states from the same set in~$X_i$ the $L_1$-distance between the successor distributions over $X_{i-1}$ is at most $\epsilon_2$. The algorithm terminates when no splitting can be done.
Let $X$ be the final partition produced by the approximate partition refinement. For any two states $s, t \in E$ where $E \in X$, we have $\ell(s) = \ell(t)$ and $\|(\tauHyp(s)(E'))_{E' \in X} - (\tauHyp(t)(E'))_{E' \in X}\|_1 \le \epsilon_2$.

Let us give more details on the loop (lines~6-14) that partitions an $E \in X_i$.
For a state $s \in E$, a candidate set $\mathit{ESet}$ is computed such that for all $E' \in \mathit{ESet}$ the state $s$ and all $x \in E'$ have the same label and the $L_1$-distance between the successor distributions over $X_{i-1}$ of $s$ and any $x \in E'$ is at most $\epsilon_2$ (line~7). If $\mathit{ESet}$ is empty, we add the singleton $\{s\}$ into $X_E$ (line~8 and~13). If there is only one set $E'$ in $\mathit{ESet}$, we add $s$ to the set $E'$. Otherwise, if there are multiple elements in $\mathit{ESet}$ that satisfy the condition, we select the one as $E'$ such that the average $L_1$-distance between the successor distributions of $s$ and $x \in E'$ is the smallest (line~10). We add $s$ to the selected set $E'$ and include $E'$ in $X_{E}$ (line~10-13).



\paragraph*{Minimisation Algorithm Using Approximate Partition Refinement}

\begin{algorithm}[h]
	\DontPrintSemicolon
	\KwIn{An LMC $\Hyp = <S, L, \tauHyp, \ell>$, a compression parameter $\epsilon_2$}
	\KwOut{An LMC $\Q_{i}$}
	$i := 0$\\
	$\Q_{i} := \M_{i}/_{\sim_{\M_{i}}}$ where $\M_{i} = \Hyp$ and $\Q_{i} = <S^{\Q_{i}}, L, \tau^{\Q_{i}}, \ell^{\Q_{i}}>$ \;
	\Repeat{$|S^{\Q_{i}}| = |S^{\Q_{i-1}}|$}{
		Compute $X_i$ by running Algorithm~\ref{alg:polynomial-optimistic-partition-refinement} with $\Q_{i} $ and $\epsilon_2$ as input\;
		Construct an LMC $\M_{i+1} :=  <X_i, L , \tau^{\M_{i+1}}, \ell^{\M_{i+1}} >$ from $\Q_i$ where $\tau^{\M_{i+1}}(E) := \sum_{u \in  E}  \frac{(\tau^{\Q_{i}}(u)(E'))_{E' \in X_i}}{|E|}$ and $\ell^{\M_{i+1}}(E) := \ell^{\Q_{i}}(x)$ for all $E \in X_i$ and any $x \in E$\;
		$\Q_{i+1} := \M_{i+1} /_{\sim_{\M_{i+1}}}$\;
		$i := i+1$ \;
	}
	\caption{LMC Minimisation by Approximate Partition Refinement}
	\label{alg:approximate-partition-refinement-merge-algorithm}
\end{algorithm}

The minimisation algorithm using approximate partition refinement is shown in \cref{alg:approximate-partition-refinement-merge-algorithm}. The input is the same as the first minimisation algorithm: an LMC $\Hyp$ and a compression parameter $\epsilon_2$. An index $i$ is initialised to $0$. Similar to the approximate minimisation algorithm using local bisimilarity distances, we also start by computing the quotient LMC $\Q_{0} = \Hyp/_{\sim_{\Hyp}}$. It then steps into a loop. We compute the approximate partition $X_i$ of $\Q_{i}$ on line~4 and construct a new LMC $\M_{i+1}$ by setting $X_i$ as its state space on line~5. For any state $E \in X_i$, we set the probability distribution as the average probability distribution over $X_i$ from all $u \in E$. The label of any $E \in X$ is set to $\ell^{\Q_{i}}(u)$ where $u$ can be any state from $E$. A new approximate quotient $\Q_{i+1} $ is obtained by taking the exact quotient of $\M_{i+1}$. We increment $i$ at the end of the iteration and continue another iteration if the size of the state space of the new approximate quotient decreases. Otherwise, the algorithm terminates as we have no states to merge. As there are finitely many states, the algorithm always terminates.


Let $i \in \nat$. The following theorem applies to both the LMCs $\Q_{i}$ from \cref{alg:local-distance-merge-algorithm} and those from \cref{alg:approximate-partition-refinement-merge-algorithm}.
\begin{restatable}{theorem}{theoremBoundGlobalDistanceApproximatePartitionRefinement}\label{theorem:bounding-global-distance}
 For all $i \in \nat$, we have that $\Q_{i+1}$ is an $\epsilon_2$-quotient of $\Q_i$. Furthermore, by the additivity lemma, we have that $\Q_i$ is an $i\epsilon_2$-quotient of $\Hyp$. 
\end{restatable}

In the case that $\Hyp=<S, L, \tauHyp, \ell>$ is a slightly perturbed version of $\M = <S, L,\tau,\ell>$, that is, for all $s \in S$ we have $\|\tau(s) - \tauHyp(s)\|_1 \le \epsilon$, the following corollary holds:
\begin{restatable}{corollary}{corollaryBoundQuotientError}\label{corollary:bounding-quotient-error}
 For all $i \in \nat$, we have that $\Q_i$ is an $(\epsilon + i\epsilon_2)$-quotient of $\M$.
\end{restatable}


\section{Active LMC Learning}
\label{section:active-LMC-learning}
We apply our approximate minimisation algorithms in a setting of active learning. Before that, we first describe how to obtain a perturbed LMC $\Hyp$ by sampling. Assume that we want to learn the transition probabilities of an LMC $\M$, that is, the state space, the labelling and the transitions are known. We also assume the system under learning (SUL) $\M$ could answer the query $\nxt$ which takes a state $s$ as input and returns a successor state of $s$ according to the transition probability distribution $\tau(s)$.

Given a state $s$ of the LMC. We denote by $x_s$ the number of successor states of $s$ and by~$n_s$ the number of times we query the SUL on $\nxt(s)$. Let $N_{s,t}$ be the frequency counts of the query result $t$, that is, the number of times a successor state $t$ appears as the result returned by the queries. We approximate the transition probability distribution by $\tauHyp(s)$ where $\tauHyp(s)(t) = \frac{N_{s, t}}{n_s}$ for all successor states $t$ of $s$. (Such an estimator is called an empirical estimator in the literature.)

Intuitively, the more queries we ask the SUL, the more accurate the approximate probability distribution $\tauHyp(s)$ would be. In fact, the following theorem holds~\cite[Section~6.4]{BazilleGenestJegourelSun2020},~\cite{Chen2015}.

\begin{theorem}\label{theorem:sampling-size} Let $\epsilon \gr 0$ be an error parameter and $\delta \gr 0$ be an error bound. Let $s \in S$. We have $\Pr(\|\tau(s) - \tauHyp(s)\|_1 \le \epsilon) \ge 1 -\delta$ for $n_s \ge \frac{1}{2 \epsilon^{2}}\ln(\frac{2x_s}{\delta})$.
\end{theorem}

For each state $s \in S$, we query the SUL on $\nxt(s)$ for $n_s \ge \frac{1}{2 \epsilon^{2}}\ln(\frac{2x_s}{\delta})$ times. We can make $\delta$ small since it appears in the logarithmic term. We then approximate the transition function by $\tauHyp$ and construct a hypothesis LMC $\Hyp = <S, L, \tauHyp, \ell>$. Since the queries $\nxt(s)$ and $\nxt(t)$ for all $s,t \in S$ and $s \not= t$ are mutually independent, by Theorem~\ref{theorem:sampling-size}, we have that $\Pr(\forall s\in S: \|\tau(s) - \tauHyp(s)\|_1 \le \epsilon) \ge (1-\delta)^{|S|} $.

We then apply the minimisation algorithms with compression parameter $\epsilon_2$ on $\Hyp$ and obtain a minimised system $\Q_{i}$ which is an $i\epsilon_2$-quotient of $\Hyp$, the LMC constructed by sampling. Since with high probability the LMC $\Hyp$ (or its exact quotient $\Q_{0}$) has small distance $\epsilon$ with the SUL $\M$, it follows from \cref{corollary:bounding-quotient-error} that with high probability the minimised system $\Q_{i}$ is an $\epsilon'$-quotient of $\M$ where $\epsilon'$ is small:  for all $i \in \nat$, we have $\Pr(\Q_i \text{ is an }\epsilon'\text{-quotient of } \M \text{ with } \epsilon' \le \epsilon + i \epsilon_2) \ge (1-\delta)^{|S|}$. The probability does not come from our minimisation algorithms and depends solely on the sampling procedure.



\section{Experiments}
\label{section:experiments}
In this section, we evaluate the performance of approximate minimisation algorithms on a number of LMCs. These LMCs model randomised algorithms and probabilistic protocols that are part of the probabilistic model checker PRISM \cite{KNP11}. The LMCs we run experiments on have less than $100,000$ states and model the following protocols or randomised algorithms: Herman's self-stabilisation algorithm \cite{Her90}, the synchronous leader election protocol by Itai and Rodeh \cite{ItaiR90}, the bounded retransmission protocol \cite{DarhenioJJL01}, the Crowds protocol \cite{ReiterR98} and the contract signing protocol by Even, Goldreich and Lempel \cite{EvenGL85}.

We implemented algorithms to obtain the slightly perturbed LMCs $\Hyp$. We call LMCs with fewer than $300$ states small; otherwise we call them large. For small LMCs, we sample the successor distribution for each state and obtain an approximation of it with error parameter $\epsilon$ and error bound $\delta$. For large LMCs, sampling is not practical as the sample size required by \cref{theorem:sampling-size} is very large. For these LMCs, we perturb the successor distribution by adding small noise to the successor transition probabilities so that for each state with at least probability $1-\delta$ the $L_1$-distance of the successor distributions in the perturbed and unperturbed systems is at most $\epsilon$ and otherwise the $L_1$-distance is $2\epsilon$. We vary the error parameter $\epsilon$ in the range of $\{0.00001, 0.0001, 0.001, 0.01\}$ and fix the error bound $\delta = 0.01$. For each unperturbed LMC and a pair of $\epsilon$ and $\delta$, we generate $5$ perturbed LMCs.

We also implemented the two minimisation algorithms in Java: \cref{alg:local-distance-merge-algorithm} and \cref{alg:approximate-partition-refinement-merge-algorithm}. The source code is publicly available\footnote{\url{https://github.com/qiyitang71/approximate-quotienting}}. We show some representative results in \cref{appendix:more-results}. The full experimental results are publicly available\footnote{\url{https://bit.ly/3vcpblY}}.

\begin{table}[t]
\begin{tabularx}{\textwidth}{@{}X@{}X@{}}
	\noindent\begin{tabular}{|c|c|c|c|}
			\hline 
			\multirow{1}{*}{\shortstack[l]{Herman5}}&
			\# states&	\# trans&	\# iter\\
			\hline 						
			$\M$ \& $\Hyp$	&		 		32	  &			244				&\\
			$\M/_{\sim_{\M}}$	&		 		4	  &			11				&\\		
			$\Hyp/_{\sim_{\Hyp}}$      & 		23  &	       167             & \\
			\hline
			\multicolumn{4}{|c|}{Perturbed LMC \#1}\\
			\hline
			\multicolumn{4}{|c|}{$\epsilon_2 = 0.00001$}\\
			\hline
			local \& apr     & 23 & 167 & 0  \\
			\hline
			\multicolumn{4}{|c|}{$\epsilon_2 = 0.0001$}\\
			\hline
			local \& apr     & 22 & 143 & 1  \\
			\hline
			\multicolumn{4}{|c|}{$\epsilon_2 \in \{0.001, 0.01, 0.1\}$}\\
			\hline
			local    & 22 & 143 & 1    \\
			\rowcolor{yellow}
			apr		   & 4 & 11 & 1   \\
			\hline
			\rowcolor{white}
			\multicolumn{4}{c}{}\\			
		\end{tabular}
	~
		\noindent\begin{tabular}{|c|c|c|c|}
		\hline
		\multirow{1}{*}{\shortstack[l]{BRP32-2
		}} &
		\# states&	\# trans&	\# iter\\
		\hline
		$\M$ \& $\Hyp$ &	1349	 &			1731    & \\
		$\M/_{\sim_{\M}}$	&		 		647	  &			903				&\\								
		$\Hyp/_{\sim_{\Hyp}}$     & 		961	   & 	   	  1343        & \\
		\hline
		\multicolumn{4}{|c|}{Perturbed LMC \#1}\\
		\hline
		\multicolumn{4}{|c|}{$\epsilon_2 = 0.00001$}\\
		\hline
		apr		&   879 & 1230 & 2 \\
		\hline
		\multicolumn{4}{|c|}{$\epsilon_2 =  0.0001$}\\
		\hline
		apr		 &  705 & 986 & 2 \\
		\hline
		\multicolumn{4}{|c|}{$\epsilon_2 \in \{0.001, 0.01\}$}\\
		\hline
		\rowcolor{yellow}
		apr		 &  647 & 903 & 1\\
		\rowcolor{white}
		\hline
		\multicolumn{4}{|c|}{$\epsilon_2 = 0.1$}\\
		\hline
		\rowcolor{red!40}
		apr		 & 196  & 387 & 1 \\
		\hline
	\end{tabular}	
\end{tabularx}
\caption{In the tables, local and apr stand for the minimisation algorithms using local bisimilarity distance and approximate partition refinement, respectively. The tables show the results for the first perturbed LMC (labeled with \#1) among the five perturbed LMCs generated by sampling or perturbing with $\epsilon =  0.0001$. (Left) Results of running the two minimisation algorithms on the LMC that models Herman's self-stabilisation algorithm with $5$ processes. (Right) Results of running apr on the LMC that models the bounded retransmission protocol with $N =32$ and $\mathit{MAX} = 2$. } \label{table:results}
\end{table}


For the small LMCs, we apply both approximate minimisation algorithms to the perturbed LMCs with $\epsilon_2 \in \{0.00001, 0.0001, 0.001, 0.01, 0.1\}$. The results for a small LMC which models the Herman's self-stabilisation algorithm is shown on the left of \cref{table:results}. For the large LMCs, we only apply the approximate minimisation algorithm using approximate partition refinement to the perturbed LMCs, since the other minimisation algorithm could not finish on the large LMCs with timeout of two hours. The results for a large LMC which models the bounded retransmission protocol is shown on the right of \cref{table:results}.   

For almost all models, given a perturbed LMC, we are able to recover the structure of the quotient of the unperturbed LMC when $\epsilon_2$ is appropriately chosen, that is, $\epsilon_2$ is no less than $\epsilon$ and is not too big; for example, see \cref{table:results} where the rows are highlighted in yellow. However, when $\epsilon_2$ is too big, the approximate minimisation algorithms may aggressively merge some states in the perturbed LMC and result in a quotient \modify{whose size is even smaller than that of the quotient of the unperturbed LMC}, as highlighted in red in \cref{table:results}. Also, we find that,  as expected, the exact partition refinement in general could not recover the structure of quotient of the original LMCs, except for the LMCs which model the synchronous leader election protocol by Itai and Rodeh. Furthermore, compared to the other approximate minimisation algorithm using the local bisimilarity distance, the one using approximate partition refinement performs much better in terms of running time and the ability to recover the structure of the quotient of the original model.


One might ask whether the minimisation algorithm using approximate partition refinement always performs better than the one using the local bisimilarity distances. \modify{In general, this is not the case as shown by \cref{example:local-merging-better}. 

\begin{example}\label{example:local-merging-better}	
	Consider the LMC $\M = <S, L, \tau, \ell>$ shown in \cref{fig:example-local-merging-better}. Let $\epsilon_2 = 0.1$. First, we run Algorithm~\ref{alg:local-distance-merge-algorithm}. It proceeds in two iterations. In the first iteration, it computes the local bisimilarity distances for all pairs of states with the same label. We have $d_{\local}(s_1, s_2) = d_{\local}(s_2, s_3)= 0.54$ and $d_{\local}(s_1, s_3) = 0.04$. It then selects the pair $s_1$ and $s_3$ of which the local bisimilarity distance is less than $\epsilon_2$ and is the smallest. These two states are merged into $s_{13}$ in the LMC shown on the left of \cref{fig:example-local-merging-better2}. In the second iteration, the only pair of states with the same label are $s_{13}$ and $s_2$. Since $d_{\local}(s_{13}, s_2) = 0.06 \le \epsilon_2$, they are merged and we arrive at the final LMC shown on the right of \cref{fig:example-local-merging-better2}.
	
	Next, we run \cref{alg:approximate-partition-refinement-merge-algorithm} with the same inputs. In the first iteration, we run approximate partition refinement on line~5 (\cref{alg:polynomial-optimistic-partition-refinement}) and present \cref{tab:example-approximate-partition-refinement} as the possible partitions of the algorithm. At the beginning of the approximate partition refinement, we have partition $X_0$ as all states are in the same set. The states are then split by the labels and we get partition $X_1$. Next, we work on the set $\{s_1, s_2, s_3\}$. Suppose that we see $s_1$ and $s_2$ before $s_3$. We have $s_1$ and $s_2$ remain together as $\|(\tau(s_1)(E))_{E \in X_1}  - (\tau(s_2)(E))_{E \in X_1}\|_1 = 0.08 \le \epsilon_2$. However, since $\|(\tau(s_3)(E))_{E \in X_1}  - (\tau(s_2)(E))_{E \in X_1}\|_1 = 0.16 \gr \epsilon_2$, we have $ESet = \emptyset$ for $s_3$ on line~9 of Algorithm~\ref{alg:polynomial-optimistic-partition-refinement} and it is split out. In the next iteration, since $\|(\tau(s_1)(E))_{E \in X_2}  - (\tau(s_2)(E))_{E \in X_2}\|_1 = 0.54 \gr \epsilon_2$, $\{s_1, s_2\}$ is split into two singleton sets. The final partition $X_3$ in which all sets are singletons suggests no merging can be done and we are left with the original LMC $\M$.
	
	This example also shows that the order of iterating through the states matters for the approximate partition refinement algorithm. Indeed, suppose we iterate though $s_1$ and $s_3$ before $s_2$ after arriving at the partition $X_1$, we will have \cref{tab:example-approximate-partition-refinement2} as the partitions and finally get the  LMC on the right of \cref{fig:example-local-merging-better2} just as the other minimisation algorithm. 
	\qed
	
\end{example}
}

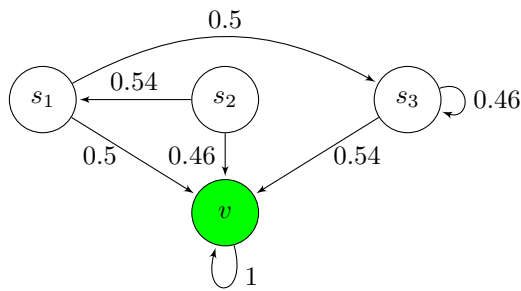
\begin{minipage}{0.5\textwidth}
	
	\begin{tikzpicture}[xscale=.6,>=latex',shorten >=1pt,node distance=3cm,on grid,auto]
	
	\node[state] (s) at (-6,0) {$s_1$};
	\node[state] (t) at (-2,0) {$s_2$};
	\node[state] (u) at (2,0) {$s_3$};
	\node[state, fill=green] (v) at (-2,-1.5) {$v$};
	
	\path[->] (s) edge [out=20,in=160] node [midway, above] {$0.5$} (u);
	\path[->] (s) edge node [midway, left, xshift=-0.1cm] {$0.5$} (v);
	
	\path[->] (t) edge node [midway, left] {$0.46$} (v);
	\path[->] (t) edge node [midway, above]  {$0.54$} (s);
	
	\path[->] (u) edge node [midway, right, xshift=0.1cm] {$0.54$} (v);
	\path[->] (u) edge [loop right] node [midway, right] {$0.46$} (u);
	
	\path[->] (v) edge [out=-70, in=-110, looseness=4] node [near start, right] {$1$} (v);
	\end{tikzpicture}
	\captionof{figure}{The LMC for which Algorithm~\ref{alg:local-distance-merge-algorithm} may perform better than Algorithm~\ref{alg:approximate-partition-refinement-merge-algorithm}.}
	\label{fig:example-local-merging-better}
\end{minipage}
\hfill
\begin{minipage}{0.45\textwidth}
	\centering	\vspace{0.1cm}
	\begin{tabular}{l}
		\\\\\\
		\hline
		$X_0 = \{S\}$ \\
		$X_1 = \big\{  \{s_1,s_2,s_3\} , \{v\} \big\}$\\
		$X_2 = \big\{  \{s_1, s_2\}, \{s_3\} , \{v\} \big\}$\\
		$X_3 = \big\{  \{s_1\},\{s_2\}, \{s_3\} , \{v\} \big\}$\\
		\hline
		\\\\
	\end{tabular}
	\captionof{table}{Example of running Algorithm~\ref{alg:polynomial-optimistic-partition-refinement} on the LMC in \cref{fig:example-local-merging-better}. (Suppose we iterate through $s_1$ and $s_2$ before $s_3$.)} \label{tab:example-approximate-partition-refinement}
\end{minipage}

\begin{minipage}{0.5\textwidth}
	\begin{tikzpicture}[xscale=.6,>=latex',shorten >=1pt,node distance=3cm,on grid,auto]
	
	\node[state] (s13) at (1.5,0) {$s_{13}$};
	\node[state] (s2) at (6.5,0) {$s_2$};
	\node[state, fill=green] (v1) at (4,-1.5) {$v$};
	
	\path[->] (s13) edge node [midway, left, xshift=-0.1cm] {$0.52$} (v1);
	\path[->] (s13) edge [loop left] node [near end, left] {$0.48$} (s13);
	
	\path[->] (s2) edge node [midway, left] {$0.46$} (v1);
	\path[->] (s2) edge node [midway, above]  {$0.54$} (s13);
	
	\path[->] (v1) edge [out=-70, in=-110, looseness=4] node [near start, right] {$1$} (v1);
	
	\node[state] (ss) at (8.5,0) {$s_{123}$};
	\node[state, fill=green] (vv) at (8.5,-1.5) {$v$};
	
	\path[->] (ss) edge node [midway, right] {$0.49$} (vv);
	\path[->] (ss) edge [loop right] node [midway, right] {$0.51$} (ss);
	\path[->] (vv) edge [out=-70, in=-110, looseness=4] node [near start, right] {$1$} (vv);
	
	\end{tikzpicture}
	\captionof{figure}{Two Steps of Running Algorithm~\ref{alg:local-distance-merge-algorithm}.}\label{fig:example-local-merging-better2}
\end{minipage}
\hfill
\begin{minipage}{0.45\textwidth}
	\centering\vspace{0.2cm}
	\begin{tabular}{l}
		\\\\
		\hline
		$X_0 = \{S\}$ \\
		$X_1 = \big\{  \{s_1,s_3,s_2\} , \{v\} \big\}$\\
		$X_2 = \big\{  \{s_1, s_3\}, \{s_2\} , \{v\} \big\}$\\
		\hline
		\\
	\end{tabular}
	\captionof{table}{Example of running Algorithm~\ref{alg:polynomial-optimistic-partition-refinement} on the LMC in \cref{fig:example-local-merging-better}. (Suppose we iterate through $s_1$ and $s_3$ before $s_2$.)} \label{tab:example-approximate-partition-refinement2}
\end{minipage}

\section{Conclusion}\label{section:conclusion}

We have developed and analysed algorithms for minimising probabilistic systems via approximate bisimulation.
These algorithms are based on $\epsilon$-quotients, a novel yet natural notion of approximate quotients.
We have obtained theoretical bounds on the discrepancy between the minimised and the non-minimised systems.
In our experiments, approximate partition refinement does well in minimising labelled Markov chains with perturbed  transition probabilities,
suggesting that approximate partition refinement is a practical approach for ``recognising'' and exploiting approximate bisimulation.

Future work might consider the following questions: Does approximate minimisation allow for further forms of active learning? Can our techniques be transferred to Markov decision processes?

\bibliography{paper}

\newpage\appendix\label{section:appendix}
\section{Proofs of \cref{section:approximate-quotient-properties}}\label{appendix: global-bisimulation-properties}

We denote by $\SubDist(S)$ the set of probability subdistributions on~$S$. The ceiling of a real number $r$, denoted by $\lceil r\rceil$, is the least integer $n$ such that $n \ge r$.

\theoremMinimumApproximateQuotientNPComplete*

\begin{proof}
	We first show that this problem is in {\sf NP}. If there is an $\errorParam$-quotient of $\Hyp =<S, L,\tauHyp,\ell>$ of size $k$ which we denote by $\Q$, by definition there is an LMC $\Hyp' = <S, L, \tau', \ell>$ such that $\Hyp'$ and $\Q$ are probabilistic bisimilar. It follows that the state space of $\Hyp'$ can be partitioned into $k$ sets and each set is a probabilistic bisimulation induced equivalence class. We summarise the idea into the following nondeterministic algorithm:  first guess a partition of the state space $S$ into $k$ sets, $E_1, \cdots, E_k$, and then check a) each subset $E_i$ is a probabilistic bisimulation induced equivalence class of $\Hyp'$; b) $\|\tau'(s)-\tauHyp(s)\|_1 \le \errorParam$ for all $s \in S$. This amounts to a feasibility test of the linear program:

	\begin{align*}
	\exists \tau': S\times S \to [0, 1] \text{ such that } 
	&\textstyle\sum_{v \in S} \tau'(u)(v) = 1 \text{ for all $u\in S$ and}\\
	&\|\tauHyp(u) -\tau'(u)\|_1 \le \errorParam \text{ for all $u \in S$ and}\\
	&\tau'(u)(E_j) = \tau'(v)(E_j) \text{ for all $u, v \in E_i$ and all $E_j$},
	\end{align*}
	and hence can be decided in polynomial time.  
	
	The Subset Sum problem is polynomial-time many-one reduction to this problem, hence it is {\sf NP}-hard. 
	
	Given an instance of Subset Sum $<P, N>$ where $P =\{p_1, \cdots, p_n\}$ and $N \in \nat$, we construct an LMC; see \cref{fig:reductionfromSubsset}. Let $T = \sum_{p_i \in P} p_i$. Let $\epsilon = \errorParam = \frac{1}{2T}$ and $k = 5$. In the LMC, state $s$ transitions to state $s_i$ with probability $p_i / T$ for all $1 \le  i \le n$. Each state $s_i$ transitions to $s_a$ and $s_b$ with equal probabilities. State $t$ transitions to $t_1$ and $t_2$ with probability $N / T$ and $1 - N / T$, respectively. State $t_1$ (resp. $t_2$) transitions to $t_a$ (resp. $t_b$) and $t_b$ (resp. $t_a$) with probability $\frac{1}{2} - \epsilon$ and $\frac{1}{2} + \epsilon$, respectively. All the remaining states transition to the successor state with probability one. States $s_b$ and $t_b$ have label $b$ and all other states have label $a$.  
	
	Next, we show that
	$$<S, N> \in {\mbox{Subset Sum}} \iff 
	\text{there is an } \epsilon\text{-quotient of }\Hyp\text{ of size is }k.$$
	
	($\implies$)
	Let $Q \subseteq P$ be the set such that $\sum_{p_i \in Q} p_i = N$. Let us define $\tau': S \to \Dist(S)$ as 
	
	$
	\tau'(u)(v) = \left \{
	\begin{array}{ll}
	\frac{(1-\epsilon)}{2} & \mbox{if $u \in Q, v = s_a$  or $u = t_1, v = t_a$}\\
	\frac{(1+\epsilon)}{2} & \mbox{if $u \in Q, v = s_b$ or $u = t_1, v = t_b$}\\
	\frac{(1+\epsilon)}{2} & \mbox{if $u \not\in Q, v = s_a$ or $u = t_2, v = t_a$}\\
	\frac{(1-\epsilon)}{2} & \mbox{if $u \not\in Q, v = s_b$ or $u = t_2, v = t_b$}\\
	\tau(u)(v) & \mbox{otherwise}\\
	\end{array}
	\right .
	$.
	
	We have $\|\tau'(u) - \tauHyp(u)\|_1 \le \epsilon$ for all $u \in S$.	Let $S_Q = \{s_i \in S \suchthat p_i \in Q\}$ and $S_{\overline{Q}} = \{ s_i \in S \suchthat p_i \not\in Q \}$. It is easy to verify that in the LMC $\Hyp' = <S,L,\tau',\ell>$, we have $s_i \sim_{\Hyp'} t_1$ for state $s_i \in S_Q$ and $s_i \sim_{\Hyp'} t_2$ for states $s_i \in S_{\overline{Q}}$. Since $\tau'(s)(S_Q) = \sum_{s_i \in S_Q} \tau'(s)(s_i) = \frac{N}{T} = \tau'(t)(t_1)$ and $\tau'(s)(S_{\overline{Q}}) = \sum_{s_i \in S_{\overline{Q}}} \tau'(s)(s_i) = 1- \frac{N}{T} = \tau'(t)(t_2)$, we have $s \sim_{\Hyp'} t$. There are five probabilistic bisimulation classes of $\Hyp'$: $\{s, t\}$, $S_Q \cup \{t_1\}$, $S_{\overline{Q}} \cup \{t_2\}$, $\{s_a, t_a\}$ and $\{s_b, t_b\}$. It follows that the exact quotient of $\Hyp'$ has five states and it is an $\epsilon$-quotient of $\Hyp$.
	
	($\impliedby$)
	Assume there is an LMC $\Q$ with five states and it is an $\epsilon$-quotient of $\Hyp$. By definition, 	there is a probabilistic transition function $\tau'$ such that $\|\tau'(s)-\tauHyp(s)\|_1 \le \epsilon$ for all $s \in S$ and the LMC $\Hyp' = <S, L, \tau', \ell>$ and $\Q$ are probabilistic bisimilar.
	
	Since any two of the five states $t$, $t_1$, $t_2$, $t_a$ and $t_b$ are not probabilistic bisimilar in $\Hyp'$, each of them will be in a different probabilistic bisimulation class of $\Hyp'$. Since $\|\tau(t_1) - \tau(t_2)\|_1 = 4\epsilon$, we have $t_1 \not\sim_{\Hyp'} t_2$. It is not hard to see that the state $s$ of $\Hyp'$ belongs to the probabilistic bisimulation class that contains $t$. Let $S_1$ be the set of states $s_i$'s such that $s_i \sim_{\Hyp'} t_1$ and $S_2$ be the set of states $s_i$'s such that $s_i \sim_{\Hyp'} t_2$. Let $Q = \{p_i \in P \suchthat s_i \in S_1 \}$. Then, $\tau'(s)(S_1) = \tau'(t)(t_1)$.
	
	Since $\|\tau'(s) - \tauHyp(s)\|_1 \le \epsilon$, we have \begin{equation}
	\label{eqn:np-eq1}
	|\tau'(s)(S_1) -\tauHyp(s)(S_1)| \ls \epsilon.
	\end{equation}
	
	It is not possible to have equality in \eqref{eqn:np-eq1}. Towards a contradiction, assume that we could have $|\tau'(s)(S_1) -\tauHyp(s)(S_1)| = \epsilon$. Without loss of generality, assume $\tau'(s)(S_1) -\tauHyp(s)(S_1) = \epsilon$. We have that 
	\begin{eqnarray*}
		&& \epsilon\\		
		&=& \tau'(s)(S_1) -\tauHyp(s)(S_1) \\
		&=& -1 + 1+ \tau'(s)(S_1) -\tauHyp(s)(S_1) \\
		&=& (\tau'(s)(S_1) -1) + (1-\tauHyp(s)(S_1)) \\
		&=& -\tau'(s)(S_2) + \tauHyp(s)(S_2) 
	\end{eqnarray*}
	Thus,  $|\tauHyp(s)(S_2)  -\tau'(s)(S_2)| = \epsilon$.  It contradicts $\|\tau'(s) -\tauHyp(s)\|_1 \le \epsilon$ as $\|\tau'(s) -\tauHyp(s)\|_1 \ge |\tau'(s) (S_1)-\tauHyp(s)(S_1)| + |\tau'(s) (S_2)-\tauHyp(s)(S_2)| = 2\epsilon$. 
	
	Similarly, since $\|\tau'(t) - \tauHyp(t)\|_1 \le \epsilon$, we have 
	\begin{equation}
	\label{eqn:np-eq2}
	|\tau'(t)(t_1) - \tauHyp(t)(t_1)| \ls \epsilon.
	\end{equation}

	Then,
	\begin{eqnarray*}
		&&0\\
		&=&	 \tau'(s)(S_1) - \tau'(t)(t_1) \\
		&=&(\tau'(s)(S_1)- \tauHyp(s)(S_1)) +(\tauHyp(t)(t_1) - \tau'(t)(t_1)) +\tauHyp(s)(S_1) - \tauHyp(t)(t_1) \\
		&\le&| \tau'(s)(S_1)-  \tauHyp(s)(S_1)| +|\tauHyp(t)(t_1) - \tau'(t)(t_1)| + \tauHyp(s)(S_1) - \tauHyp(t)(t_1) \\
		&\ls& \epsilon + \epsilon + \tauHyp(s)(S_1) - \tauHyp(t)(t_1) \commenteq{\eqref{eqn:np-eq1} and \eqref{eqn:np-eq2}}\\
		&=&2\epsilon + \tauHyp(s)(S_1) - \tauHyp(t)(t_1) 
	\end{eqnarray*}
	
	and 
	\begin{eqnarray*}
		&&0\\
		&=&	\tau'(s)(S_1) - \tau'(t)(t_1) \\
		&=&(\tau'(s)(S_1)-  \tauHyp(s)(S_1)) +(\tauHyp(t)(t_1) - \tau'(t)(t_1)) + \tauHyp(s)(S_1) - \tauHyp(t)(t_1) \\
		&\ge&-|\tau'(s)(S_1)- \tauHyp(s)(S_1)| -|\tauHyp(t)(t_1) - \tau'(t)(t_1)| +  \tauHyp(s)(S_1) - \tauHyp(t)(t_1) \\
		&\gr& -\epsilon - \epsilon + \tauHyp(s)(S_1) - \tauHyp(t)(t_1) \commenteq{\eqref{eqn:np-eq1} and \eqref{eqn:np-eq2}}\\
		&=&-2\epsilon +  \tauHyp(s)(S_1) - \tauHyp(t)(t_1) .
	\end{eqnarray*}
	
	From the above, we have 
	\begin{eqnarray*}
		&&-2\epsilon \ls\tauHyp(s)(S_1) - \tauHyp(t)(t_1) \ls 2\epsilon\\
		&\iff& -\frac{1}{T} \ls \tauHyp(s)(S_1) - \tauHyp(t)(t_1) \ls \frac{1}{T} \commenteq{$\epsilon = \frac{1}{2T}$}\\
		&\iff& -\frac{1}{T} \ls \sum_{p_i \in Q} \frac{p_i}{T} - \frac{N}{T} \ls \frac{1}{T}\\
		&\iff& N-1\ls \sum_{p_i \in Q} p_i \ls N+1
	\end{eqnarray*}
	
	Since $\sum_{p_i \in Q} p_i$ is an integer,  it follows that $\sum_{p_i \in Q} p_i = N$.
\end{proof}


\propositionApproximateGlobalRelationSubset*

\begin{proof}
	Let $\Q$ be an $\errorParam$-quotient of $\Hyp$. By definition, there is a probability transition function $\tau'$ with $\|\tauHyp(s) - \tau'(s)\|_1 \le \errorParam$ for all $s\in S$ such that $\Q$ is the exact quotient of the LMC $\Hyp' = <S, L, \tau', \ell>$. Let $s$ be an arbitrary state from $\Hyp$. Let $s'$ denote the corresponding state from $\Hyp'$ and $s^{Q} = [s]^{\errorParam}$ the corresponding state from $\Q$.  To prove $s \sim_{\frac{\errorParam}{2}} s^{Q}$, it suffices to prove $s \sim_{\frac{\errorParam}{2}} s'$ since $s' \sim s^{Q}$ and approximate bisimulation satisfies the additivity property~\cite{DesharnaisLavoletteTracol2008}.  
	
	Next, we prove that the following relation $\mathcal{R} = \{(s, s) \suchthat \text{ the first } s \text{ is from } \Hyp \text{ and the second } s \text{ is the corresponding state from } \Hyp'\}$ is an $\frac{\errorParam}{2}$-bisimulation. It is obvious that both states $s$ have the same label and it remains to show that $(\tauHyp(s), \tau'(s)) \in \lifting{\mathcal{R}}_{\frac{\errorParam}{2}}$. To do that, we construct an $\omega \in \Omega(\tauHyp(s), \tau'(s))$ such that $\sum_{(u, v) \in \mathcal{R}} \omega(u, v) \ge 1 -\frac{\errorParam}{2}$. If $\|\tauHyp(s) - \tau'(s)\|_1 = 0$, we can define $\omega(u, u) = \tauHyp(s)(u) = \tau'(s)(u)$ and thus, $\sum_{(u, u) \in \mathcal{R}} \omega(u, u) = \sum_{u \in S}\tauHyp(s)(u) = 1 \ge 1 -\frac{\errorParam}{2}$. For the remainder of this proof, we assume $\|\tauHyp(s) - \tau'(s)\|_1 \gr 0$.
	
	Define $\omega'$ such that $\omega'(u, u) = \min\{\tauHyp(s)(u), \tau'(s)(u)\}$ for all $u \in S$. Let $\alpha \in \SubDist(S)$ such that $\alpha(u) = \max\{\tauHyp(s)(u)-\tau'(s)(u) ,0\}$ where $u \in S$. We have $\|\alpha\|_1 = \sum_{u \in S} \alpha(u) = \sum_{u \in S}  \max\{\tauHyp(s)(u)-\tau'(s)(u), 0\} =\frac{1}{2} \|\tauHyp(s)-\tau'(s)\|_1$. Let $\beta \in \SubDist(S)$ such that $\beta(u) = \max\{\tau'(s)(u)-\tauHyp(s)(u) ,0\}$. Similarly, we have $\|\beta\|_1 = \frac{1}{2} \|\tauHyp(s)-\tau'(s)\|_1$.
	
	Since $\|\alpha\|_1 = \|\beta\|_1 = \frac{1}{2} \|\tauHyp(s)-\tau'(s)\|_1 \gr 0$, there exists a coupling $\gamma \in \Omega(\alpha/\|\alpha\|_1, \beta/ \|\beta\|_1)$. Next, we show that $\omega'+\|\alpha\|_1\gamma \in \Omega(\tauHyp(s), \tau'(s))$. We check the conditions on the left and right marginals, respectively:
	\begin{align*}
	& \phantom{=} \textstyle\sum_{v \in S} (\omega'+\|\alpha\|_1\gamma)(u, v) \\
	&= \textstyle\sum_{v \in S} \omega'(u, v)+\sum_{v \in S} \|\alpha\|_1\gamma(u, v)\\
	&= \textstyle\min\{\tauHyp(s)(u), \tau'(s)(u)\} + \|\alpha\|_1\frac{\alpha(u)}{\|\alpha\|_1} \commenteq{$\sum_{v \in S} \gamma(u, v) = \alpha(u) / \|\alpha\|_1$}\\
	&= \textstyle\min\{\tauHyp(s)(u), \tau'(s)(u)\}  +  \max\{\tauHyp(s)(u)-\tau'(s)(u), 0\} \\
	&= \tauHyp(s)(u) \text{ and }\\
	&\phantom{=} \textstyle\sum_{u\in S} (\omega'+\|\alpha\|_1\gamma)(u, v) \\
	&=\textstyle\sum_{u\in S} (\omega'+\|\beta\|_1\gamma)(u, v) \commenteq{$\|\alpha\|_1 = \|\beta\|_1$}\\
	&= \textstyle\sum_{u \in S} \omega'(u, v)+\sum_{u \in S} \|\beta\|_1\gamma(u, v) \\
	&= \textstyle\min\{\tauHyp(s)(v), \tau'(s)(v)\}  +  \max\{\tau'(s)(v)-\tauHyp(s)(v), 0\} \\
	&= \tau'(s)(v).
	\end{align*}
	
	Finally, we check that the constructed coupling $\omega = \omega'+\|\alpha\|_1\gamma$ satisfies that $\textstyle\sum_{(u, v) \in \mathcal{R}} \omega(u, v) \ge 1 - \frac{\errorParam}{2}$: 
	\begin{align*}
	\textstyle\sum_{(u, v)\in \mathcal{R}}\omega (u, v) &=\textstyle\sum_{(u, u) \in \mathcal{R}} (\omega'+\|\alpha\|_1\gamma)(u, u) \\
	&\ge \textstyle\sum_{(u, u) \in \mathcal{R}} \omega'(u, u) \\
	&= \textstyle\sum_{(u, u) \in \mathcal{R}} \min\{\tauHyp(s)(u), \tau'(s)(u)\} \\	
	&=\textstyle\sum_{u \in S} \min\{\tauHyp(s)(u), \tau'(s)(u)\}\\
	&=\textstyle\sum_{u \in S} \big( \tauHyp(s)(u) - \max\{\tauHyp(s)(u) - \tau'(s)(u), 0\} \big)\\
	&=\textstyle\sum_{u \in S} \tauHyp(s)(u) - \sum_{u \in S} \max\{\tauHyp(s)(u) - \tau'(s)(u), 0\}\\	
	&=1 - \textstyle\sum_{u \in S} \max\{\tauHyp(s)(u) - \tau'(s)(u), 0\}\\	
	&=1 - \frac{1}{2} \|\tauHyp(s)-\tau'(s)\|_1\\		
	&\ge 1 - \frac{\errorParam}{2}\qedhere
	\end{align*}
\end{proof}


\begin{example}\label{example:s-t-not-global-2epsilon-related}
Consider the LMC in \cref{fig:example-not-subseteq-R-2-epsilon}. We have that $t_1 \sim_{\epsilon}  s  \sim_{\epsilon}  t \sim_{\epsilon}  s_1$. It is not hard to see that $t_1 \sim_{\epsilon}  s$ and $t \sim_{\epsilon}  s_1$. Let $\R$ be a reflexive and symmetric relation such that $\{(s, t), (s_1, t), (s, t_1)\} \subseteq \R$ and $(s_1, t_1) \not\in \R$. To show $s \sim_{\epsilon}  t$, it suffices to show that $(\tau(s), \tau(t)) \in \lifting{\R}_{\epsilon}$; more precisely, a coupling $\omega \in \Omega(\tau(s), \tau(t))$ such that $\sum_{ (u,v) \in \R} \omega(u, v) \ge 1 - \epsilon$. Let $\omega(s, t) = \omega(s_1, t) = \omega(s, t_1) = \frac{1}{4}$, $\omega(x,x) = \frac{1}{4} - \epsilon$, $\omega(x, t) = \epsilon$ and $\omega(u,v) = 0$ for all other $u, v \in S$. It is easy to check that $\omega \in \Omega(\tau(s), \tau(t))$ and it satisfies that  $\sum_{ (u,v) \in \mathcal{R}} \omega(u, v) \ge 1 - \epsilon$. Next, we show the resulting LMC in \cref{fig:exampleMerge2} which is obtained by merging the states that are related by the transitive closure of $\sim_{\epsilon}$. This LMC is a $3\epsilon$-quotient of the LMC in \cref{fig:example-not-subseteq-R-2-epsilon}.

\begin{figure}[h]
	\begin{tikzpicture}
	\tikzstyle{BoxStyle} = [draw, circle, fill=black, scale=0.05,minimum width = 0.001pt, minimum height = 0.001pt]
	\node[state] (s) at (1,4) {$s$};
	\node[state] (s1) at (0,2.5) {$s_{1}$};
	\node[state, fill=green] (x) at (2,2.5) {$x$};
	\path[->] (s) edge [loop left] node  [midway,left] {$\frac{1}{2}$} (s);
	\path[->] (s) edge node  [midway,left,xshift=-0.1cm,yshift=0.1cm]  {$\frac{1}{4}$} (s1);
	\path[->] (s) edge node  [midway,right,xshift=0.1cm,yshift=0.1cm] {$\frac{1}{4}$} (x);
	\path[->] (x) edge [out=-70, in=-110, looseness=4] node [near start, right] {$1$} (x);
	
	\node[state] (t) at (4.35,4) {$t$};
	\node[state] (t1) at (3.35,2.5) {$t_{1}$};
	\node[state, fill=green] (y) at (5.35,2.5) {$x$};
	\path[->] (t) edge [loop left] node  [midway,left] {$\frac{1}{2}+\epsilon$} (t);
	\path[->] (t) edge node  [midway,left,xshift=-0.1cm,yshift=0.1cm] {$\frac{1}{4}$} (t1);
	\path[->] (t) edge node  [midway,right,xshift=0.1cm,yshift=0.1cm] {$\frac{1}{4}-\epsilon$} (y);
	\path[->] (y) edge[out=-70, in=-110, looseness=4] node [near start, right] {$1$} (y);
	
	\node[state] (s1) at (11.05,4) {$s_1$};
	\node[state] (t1) at (10.05,2.5) {$t_{1}$};
	\node[state, fill=green] (x) at (12.05,2.5) {$x$};
	\path[->] (s1) edge [loop left] node  [midway,left] {$\frac{1}{2}+2\epsilon$} (s1);
	\path[->] (s1) edge node  [midway,left,xshift=-0.1cm,yshift=0.1cm] {$\frac{1}{4}$} (t1);
	\path[->] (s1) edge node  [midway,right,xshift=0.1cm,yshift=0.1cm] {$\frac{1}{4}-2\epsilon$} (x);
	\path[->] (x) edge  [out=-70, in=-110, looseness=4] node [near start, right]{$1$} (x);
	
	\node[state] (t1) at (7.7,4) {$t_1$};
	\node[state] (ts1) at (6.7,2.5) {$s_{1}$};
	\node[state, fill=green] (y) at (8.7,2.5) {$x$};
	\path[->] (t1) edge [loop left] node  [midway,left] {$\frac{1}{2}-\epsilon$} (t1);
	\path[->] (t1) edge node  [midway,left,xshift=-0.1cm,yshift=0.1cm] {$\frac{1}{4}$} (ts1);
	\path[->] (t1) edge node  [midway,right,xshift=0.1cm,yshift=0.1cm] {$\frac{1}{4}+\epsilon$} (y);
	\path[->] (y) edge  [out=-70, in=-110, looseness=4] node [near start, right] {$1$} (y);
	\end{tikzpicture}
	
	\caption{An LMC in which we have $t_1 \sim_{\epsilon} s \sim_{\epsilon} t \sim_{\epsilon} s_1$ }
	\label{fig:example-not-subseteq-R-2-epsilon}
\end{figure}
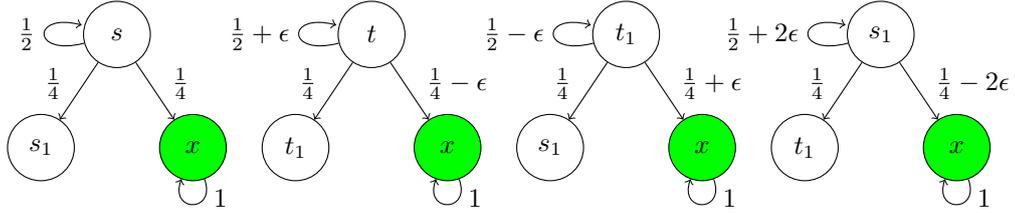

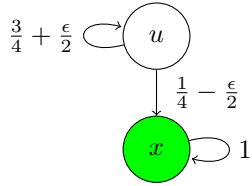
\begin{figure}[h]
	\centering
	\begin{tikzpicture}
	\tikzstyle{BoxStyle} = [draw, circle, fill=black, scale=0.05,minimum width = 0.001pt, minimum height = 0.001pt]
	
	\node[state] (t) at (12,4) {$u$};
	\node[state, fill=green] (y) at (12,2.5) {$x$};
	\path[->] (t) edge [loop left] node  [midway,left] {$\frac{3}{4}+\frac{\epsilon}{2}$} (t);
	\path[->] (t) edge node  [midway,right,xshift=0.1cm] {$\frac{1}{4}-\frac{\epsilon}{2}$} (y);
	\path[->] (y) edge [loop right] node [midway, right] {$1$} (y);
	\end{tikzpicture}
	\caption{An $\epsilon'$-quotient ($\epsilon'$ is at least $3\epsilon$) obtained by merging the states that are related by the transitive closure of $\sim_{\epsilon}$: $s$, $s_1$, $t$ and $t_1$.}
	\label{fig:exampleMerge2} 
\end{figure}

\end{example}


\begin{restatable}{theorem}{theoremGlobalBisimulationStronger}\label{theorem:global-bisimulation-is-stronger}
	Let $n \in \integer^{+}$ and $\epsilon \in (0,  \frac{1}{(\lceil\frac{n}{2}\rceil +1)2^{\lceil\frac{n}{2}\rceil +1}}]$. There is an LMC $\M(n)$ such that merging the states that are related by the transitive closure of $\sim_{\epsilon}$ results in an $\epsilon'$-quotient where $\epsilon'$ is at least $(n+1)\epsilon$. 
\end{restatable}
\begin{proof}
	Let  $n \in \integer^{+}$ and $\epsilon \in (0, \frac{1}{(\lceil\frac{n}{2}\rceil +1)2^{\lceil\frac{n}{2}\rceil +1}}]$. For LMCs $\M(2n-1)$ and $\M(2n)$, we have $\epsilon \in (0, \frac{1}{(n+1)2^{n+1}}]$.
	
	In \cref{fig:example-not-subseteq-R-n-epsilon-odd}, we have $s \sim_{\epsilon} t$. There are $2n+2$ states in total: $s$, $t$, $s_i$ where $i\in\{1,\ldots,n\}$, $t_i$ where $i\in \{1,\ldots,(n-1)\}$ and $x$. All states have the same label but state $x$. State $s$ transitions to $s_i$ with probability $\frac{1}{2^{i+1}}$ where $i \in \{1,\ldots,n\}$. It transitions back to itself with probability $\frac{1}{2}$ and to state $x$ with probability $\frac{1}{2^{n+1}}$. State $t$ transitions to $t_i$ with probability $\frac{1}{2^{i+1}}$ where $i \in \{1,\ldots,(n-2)\}$ and to $t_{n-1}$ with probability $\frac{3}{2^{n+1}}$. It transitions back to itself with probability $\frac{1}{2}+\epsilon$ and to state $x$ with probability $\frac{1}{2^{n+1}}-\epsilon$.  If $n$ is an even number, we have $s_{n-1} \sim_{\epsilon}  t_{n-2}  \cdots s_2 \sim_{\epsilon}  t_1 \sim_{\epsilon}  s \sim_{\epsilon} t \sim_{\epsilon}  s_1 \sim_{\epsilon}  t_2 \cdots  t_{n-1} \sim_{\epsilon} s_{n}$. If $n$ is an odd number, we have $s_{n} \sim_{\epsilon}  t_{n-1}  \cdots s_2 \sim_{\epsilon}  t_1 \sim_{\epsilon}  s \sim_{\epsilon} t \sim_{\epsilon}  s_1 \sim_{\epsilon}  t_2 \cdots  t_{n-2} \sim_{\epsilon} s_{n-1}$. The  LMC obtained by merging the states that are related by the transitive closure of $\sim_{\epsilon}$ is shown in \cref{fig:exampleMerge2n}(a). This LMC is a $\epsilon'$-quotient of the LMC in \cref{fig:example-not-subseteq-R-n-epsilon-odd} where $\epsilon'$ is at least $2n\epsilon$.
	

	In \cref{fig:example-not-subseteq-R-n-epsilon-even}, we have $s \sim_{\epsilon} t$. There are $2n+3$ states in total: $s$, $t$, $s_i$ where $i\in\{1,\ldots,n\}$, $t_i$ where $i\in \{1,\ldots,n\}$ and $x$. All states have the same label but state $x$. State $s$ transitions to $s_i$ with probability $\frac{1}{2^{i+1}}$ where $i \in \{1,\ldots,n\}$. It transitions back to itself with probability $\frac{1}{2}$ and to state $x$ with probability $\frac{1}{2^{n+1}}$. State $t$ transitions to $t_i$ with probability $\frac{1}{2^{i+1}}$ where $i \in \{1,\ldots,n\}$. It transitions back to itself with probability $\frac{1}{2}+\epsilon$ and to state $x$ with probability $\frac{1}{2^{n+1}}-\epsilon$.  If $n$ is an even number, we have $s_{n} \sim_{\epsilon}  t_{n-1}  \cdots s_2 \sim_{\epsilon}  t_1 \sim_{\epsilon}  s \sim_{\epsilon} t \sim_{\epsilon}  s_1 \sim_{\epsilon}  t_2 \cdots  s_{n-1} \sim_{\epsilon} t_{n}$. If $n$ is an odd number, we have $t_{n} \sim_{\epsilon}  s_{n-1}  \cdots s_2 \sim_{\epsilon}  t_1 \sim_{\epsilon}  s \sim_{\epsilon} t \sim_{\epsilon}  s_1 \sim_{\epsilon}  t_2 \cdots  t_{n-1} \sim_{\epsilon} s_{n}$. The  LMC obtained by merging the states that are related by the transitive closure of $\sim_{\epsilon}$ is shown in \cref{fig:exampleMerge2n}(b). This LMC is a $\epsilon'$-quotient of the LMC in \cref{fig:example-not-subseteq-R-n-epsilon-even} where $\epsilon'$ is at least $(2n+1)\epsilon$.
	
\end{proof}

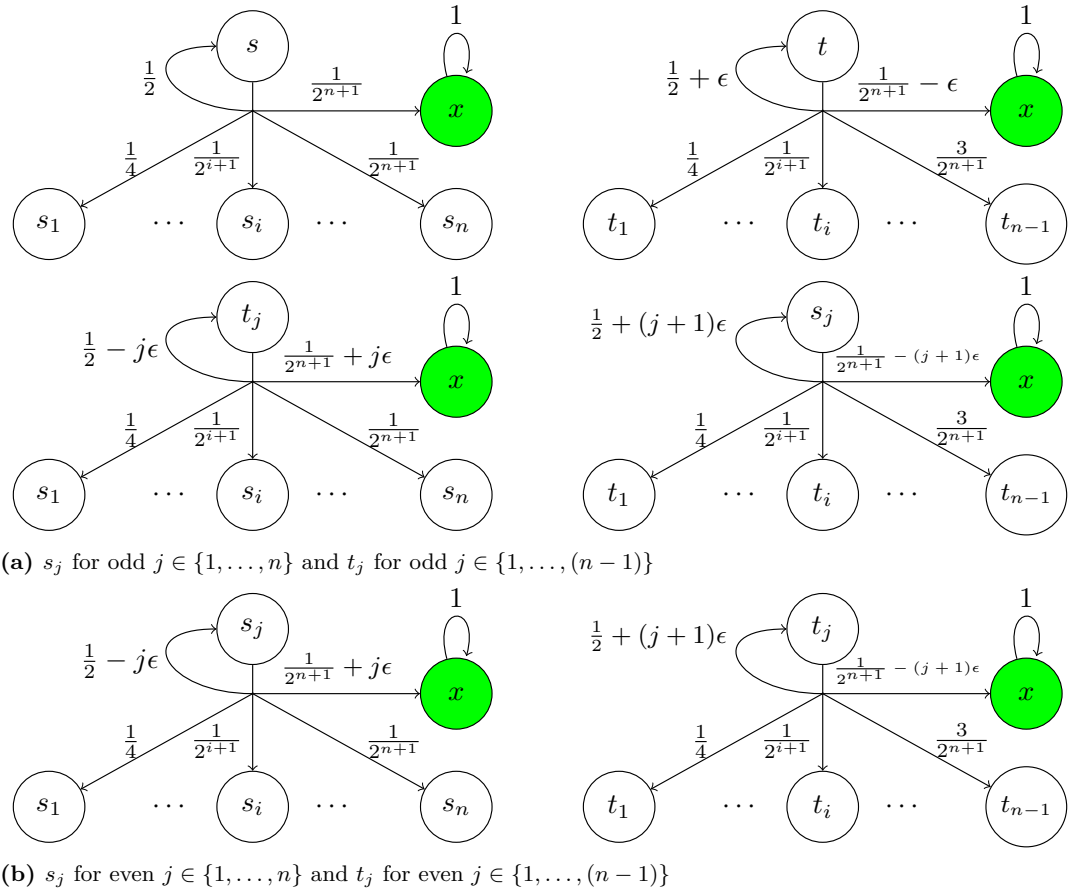
\begin{figure}[t]
		\tikzstyle{BoxStyle} = [draw, circle, fill=black, scale=0.05,minimum width = 0.001pt, minimum height = 0.001pt]
	\begin{subfigure}[h]{\textwidth}
		\resizebox{\textwidth}{!}{
			\begin{tikzpicture}              
			\node[state] (us) at (3,4) {$s$};
			\node[BoxStyle] (st) at (3,3.2){};
			\node[state] (u1) at (0.5,1.8) {$s_{1}$};
			\node at (2,1.8) {$\cdots$};	
			\node[state] (u2) at (3,1.8) {$s_{i}$};
			\node at (4,1.8) {$\cdots$};	
			\node[state] (u3) at (5.5,1.8) {$s_{n}$};	
			\node[state, fill=green] (un) at (5.5,3.2) {$x$};
			\path[-] (us) edge node  [near end,left] {} (st);
			\path[->] (st) edge [out=180,in=180,looseness=3] node  [midway,left] {$\frac{1}{2}$} (us);
			\path[->] (st) edge node  [midway,left, xshift=-0.2cm] {$\frac{1}{4}$} (u1);
			\path[->] (st) edge node  [midway,left,yshift=-0.1cm] {$\frac{1}{2^{i+1}}$} (u2);
			\path[->] (st) edge node  [midway,right, xshift=0.2cm] {$\frac{1}{2^{n+1}}$} (u3);
			\path[->] (st) edge node  [midway,above] {$\frac{1}{2^{n+1}}$} (un);
			\path[->] (un) edge [loop above] node [midway, above] {$1$} (un);
			
			\node[state] (ut) at (10,4) {$t$};
			\node[BoxStyle] (tt) at (10,3.2){};
			\node[state] (t2) at (12.5,1.8) {\small $ t_{n-1}$};
			\node at (11,1.8) {$\cdots$};
			\node[state] (ti) at (10,1.8) {$t_{i}$};
			\node at (9,1.8) {$\cdots$};
			\node[state] (t1) at (7.5,1.8) {$t_{1}$};
			\node[state, fill=green] (tn) at (12.5,3.2) {$x$};
			\path[-] (ut) edge node  [near end,left] {} (tt);
			\path[->] (tt) edge [out=180,in=180,looseness=3] node  [midway,left] {$\frac{1}{2} +\epsilon$} (ut);
			\path[->] (tt) edge node  [midway,left, xshift=-0.3cm] {$\frac{1}{4}$} (t1);
			\path[->] (tt) edge node  [midway,left, yshift=-0.1cm] {$\frac{1}{2^{i+1}}$} (ti);	
			\path[->] (tt) edge node  [midway,right, xshift=0.2cm] {$\frac{3}{2^{n+1}}$} (t2);
			\path[->] (tt) edge node  [midway,above] {$\frac{1}{2^{n+1}}-\epsilon$} (tn);
			\path[->] (tn) edge [loop above] node [midway, above] {$1$} (tn);
			\end{tikzpicture}}
	\end{subfigure}
	\begin{subfigure}[h]{\textwidth}
		\resizebox{\textwidth}{!}{
			\begin{tikzpicture}
			\node[state] (us) at (3,4) {$t_j$};
			\node[BoxStyle] (st) at (3,3.2){};
			\node[state] (u1) at (0.5,1.8) {$s_{1}$};
			\node at (2,1.8) {$\cdots$};	
			\node[state] (u2) at (3,1.8) {$s_{i}$};
			\node at (4,1.8) {$\cdots$};	
			\node[state] (u3) at (5.5,1.8) {$s_{n}$};	
			\node[state, fill=green] (un) at (5.5,3.2) {$x$};
			\path[-] (us) edge node  [near end,left] {} (st);
			\path[->] (st) edge [out=180,in=180,looseness=3] node  [midway,left] {$\frac{1}{2}-j\epsilon$} (us);
			\path[->] (st) edge node  [midway,left, xshift=-0.2cm] {$\frac{1}{4}$} (u1);
			\path[->] (st) edge node  [midway,left,yshift=-0.1cm] {$\frac{1}{2^{i+1}}$} (u2);
			\path[->] (st) edge node  [midway,right, xshift=0.2cm] {$\frac{1}{2^{n+1}}$} (u3);
			\path[->] (st) edge node  [midway,above] {\small$\frac{1}{2^{n+1}}+j\epsilon$} (un);
			\path[->] (un) edge [loop above] node [midway, above] {$1$} (un);
			
			\node[state] (ut) at (10,4) {$s_j$};
			\node[BoxStyle] (tt) at (10,3.2){};
			\node[state] (t2) at (12.5,1.8) {\small $ t_{n-1}$};
			\node at (11,1.8) {$\cdots$};
			\node[state] (ti) at (10,1.8) {$t_{i}$};
			\node at (9,1.8) {$\cdots$};
			\node[state] (t1) at (7.5,1.8) {$t_{1}$};
			\node[state, fill=green] (tn) at (12.5,3.2) {$x$};
			\path[-] (ut) edge node  [near end,left] {} (tt);
			\path[->] (tt) edge [out=180,in=180,looseness=3] node  [midway,above,xshift=-1cm] {\small $\frac{1}{2}+(j+1)\epsilon$} (ut);
			\path[->] (tt) edge node  [midway,left, xshift=-0.2cm] {$\frac{1}{4}$} (t1);
			\path[->] (tt) edge node  [midway,left, yshift=-0.1cm] {$\frac{1}{2^{i+1}}$} (ti);	
			\path[->] (tt) edge node  [midway,right, xshift=0.2cm] {$\frac{3}{2^{n+1}}$} (t2);
			\path[->] (tt) edge node  [midway,above] {\tiny $\frac{1}{2^{n+1}}-(j+1)\epsilon$} (tn);
			\path[->] (tn) edge [loop above] node [midway, above] {$1$} (tn);
			\end{tikzpicture}} \caption{$s_j$ for odd $j \in \{1,\ldots,n\}$ and $t_j$ for odd $j \in \{1,\ldots,(n-1)\}$}\label{fig:odd-n-odd-j}
	\end{subfigure}
	\begin{subfigure}[h]{\textwidth}
		\resizebox{\textwidth}{!}{
			\begin{tikzpicture}
			\node[state] (us) at (3,4) {$s_j$};
			\node[BoxStyle] (st) at (3,3.2){};
			\node[state] (u1) at (0.5,1.8) {$s_{1}$};
			\node at (2,1.8) {$\cdots$};	
			\node[state] (u2) at (3,1.8) {$s_{i}$};
			\node at (4,1.8) {$\cdots$};	
			\node[state] (u3) at (5.5,1.8) {$s_{n}$};	
			\node[state, fill=green] (un) at (5.5,3.2) {$x$};
			\path[-] (us) edge node  [near end,left] {} (st);
			\path[->] (st) edge [out=180,in=180,looseness=3] node  [midway,left] {$\frac{1}{2}-j\epsilon$} (us);
			\path[->] (st) edge node  [midway,left, xshift=-0.2cm] {$\frac{1}{4}$} (u1);
			\path[->] (st) edge node  [midway,left,yshift=-0.1cm] {$\frac{1}{2^{i+1}}$} (u2);
			\path[->] (st) edge node  [midway,right, xshift=0.2cm] {$\frac{1}{2^{n+1}}$} (u3);
			\path[->] (st) edge node  [midway,above] {\small$\frac{1}{2^{n+1}}+j\epsilon$} (un);
			\path[->] (un) edge [loop above] node [midway, above] {$1$} (un);
			
			\node[state] (ut) at (10,4) {$t_j$};
			\node[BoxStyle] (tt) at (10,3.2){};
			\node[state] (t2) at (12.5,1.8) {\small $ t_{n-1}$};
			\node at (11,1.8) {$\cdots$};
			\node[state] (ti) at (10,1.8) {$t_{i}$};
			\node at (9,1.8) {$\cdots$};
			\node[state] (t1) at (7.5,1.8) {$t_{1}$};
			\node[state, fill=green] (tn) at (12.5,3.2) {$x$};
			\path[-] (ut) edge node  [near end,left] {} (tt);
			\path[->] (tt) edge [out=180,in=180,looseness=3] node  [midway,above,xshift=-1cm] {\small $\frac{1}{2}+(j+1)\epsilon$} (ut);
			\path[->] (tt) edge node  [midway,left, xshift=-0.2cm] {$\frac{1}{4}$} (t1);
			\path[->] (tt) edge node  [midway,left, yshift=-0.1cm] {$\frac{1}{2^{i+1}}$} (ti);	
			\path[->] (tt) edge node  [midway,right, xshift=0.2cm] {$\frac{3}{2^{n+1}}$} (t2);
			\path[->] (tt) edge node  [midway,above] {\tiny $\frac{1}{2^{n+1}}-(j+1)\epsilon$} (tn);
			\path[->] (tn) edge [loop above] node [midway, above] {$1$} (tn);
			\end{tikzpicture}} \caption{$s_j$ for even $j \in \{1,\ldots,n\}$ and $t_j$ for even $j \in \{1,\ldots,(n-1)\}$}\label{fig:odd-n-even-j}
	\end{subfigure}
	\caption{The LMC $\M(2n-1)$ in which $s \sim_{\epsilon} t$. }\label{fig:example-not-subseteq-R-n-epsilon-odd}
\end{figure}

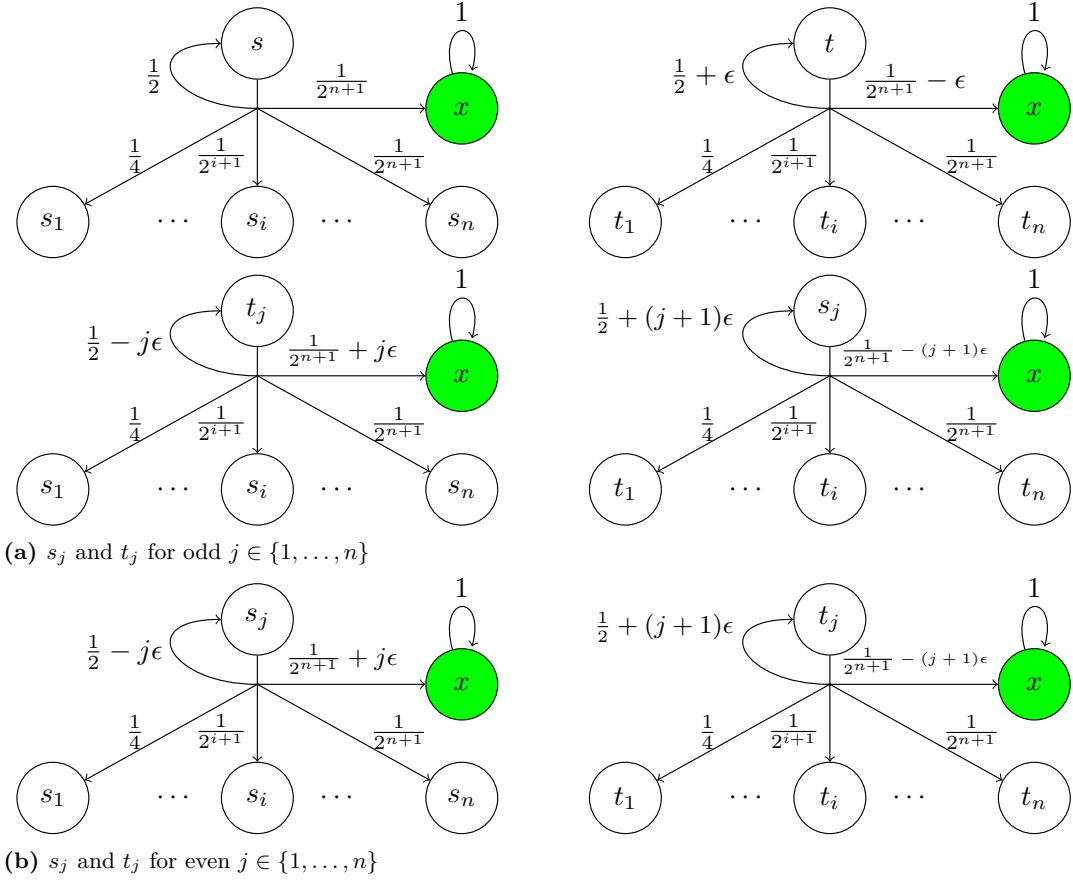
\begin{figure}[t]
	\tikzstyle{BoxStyle} = [draw, circle, fill=black, scale=0.05,minimum width = 0.001pt, minimum height = 0.001pt]
	\begin{subfigure}[h]{\textwidth}
		\resizebox{\textwidth}{!}{
			\begin{tikzpicture}              
			\node[state] (us) at (3,4) {$s$};
			\node[BoxStyle] (st) at (3,3.2){};
			\node[state] (u1) at (0.5,1.8) {$s_{1}$};
			\node at (2,1.8) {$\cdots$};	
			\node[state] (u2) at (3,1.8) {$s_{i}$};
			\node at (4,1.8) {$\cdots$};	
			\node[state] (u3) at (5.5,1.8) {$s_{n}$};	
			\node[state, fill=green] (un) at (5.5,3.2) {$x$};
			\path[-] (us) edge node  [near end,left] {} (st);
			\path[->] (st) edge [out=180,in=180,looseness=3] node  [midway,left] {$\frac{1}{2}$} (us);
			\path[->] (st) edge node  [midway,left, xshift=-0.2cm] {$\frac{1}{4}$} (u1);
			\path[->] (st) edge node  [midway,left,yshift=-0.1cm] {$\frac{1}{2^{i+1}}$} (u2);
			\path[->] (st) edge node  [midway,right, xshift=0.2cm] {$\frac{1}{2^{n+1}}$} (u3);
			\path[->] (st) edge node  [midway,above] {$\frac{1}{2^{n+1}}$} (un);
			\path[->] (un) edge [loop above] node [midway, above] {$1$} (un);
			
			\node[state] (ut) at (10,4) {$t$};
			\node[BoxStyle] (tt) at (10,3.2){};
			\node[state] (t2) at (12.5,1.8) {$ t_{n}$};
			\node at (11,1.8) {$\cdots$};
			\node[state] (ti) at (10,1.8) {$t_{i}$};
			\node at (9,1.8) {$\cdots$};
			\node[state] (t1) at (7.5,1.8) {$t_{1}$};
			\node[state, fill=green] (tn) at (12.5,3.2) {$x$};
			\path[-] (ut) edge node  [near end,left] {} (tt);
			\path[->] (tt) edge [out=180,in=180,looseness=3] node  [midway,left] {$\frac{1}{2}+\epsilon$} (ut);
			\path[->] (tt) edge node  [midway,left, xshift=-0.2cm] {$\frac{1}{4}$} (t1);
			\path[->] (tt) edge node  [midway,left, yshift=-0.1cm] {$\frac{1}{2^{i+1}}$} (ti);	
			\path[->] (tt) edge node  [midway,right, xshift=0.2cm] {$\frac{1}{2^{n+1}}$} (t2);
			\path[->] (tt) edge node  [midway,above] {$\frac{1}{2^{n+1}}-\epsilon$} (tn);
			\path[->] (tn) edge [loop above] node [midway, above] {$1$} (tn);
			\end{tikzpicture}}
	\end{subfigure}
	\begin{subfigure}[h]{\textwidth}
		\resizebox{\textwidth}{!}{
			\begin{tikzpicture}
			\node[state] (us) at (3,4) {$t_j$};
			\node[BoxStyle] (st) at (3,3.2){};
			\node[state] (u1) at (0.5,1.8) {$s_{1}$};
			\node at (2,1.8) {$\cdots$};	
			\node[state] (u2) at (3,1.8) {$s_{i}$};
			\node at (4,1.8) {$\cdots$};	
			\node[state] (u3) at (5.5,1.8) {$s_{n}$};	
			\node[state, fill=green] (un) at (5.5,3.2) {$x$};
			\path[-] (us) edge node  [near end,left] {} (st);
			\path[->] (st) edge [out=180,in=180,looseness=3] node  [midway,left] {$\frac{1}{2}-j\epsilon$} (us);
			\path[->] (st) edge node  [midway,left, xshift=-0.2cm] {$\frac{1}{4}$} (u1);
			\path[->] (st) edge node  [midway,left,yshift=-0.1cm] {$\frac{1}{2^{i+1}}$} (u2);
			\path[->] (st) edge node  [midway,right, xshift=0.2cm] {$\frac{1}{2^{n+1}}$} (u3);
			\path[->] (st) edge node  [midway,above] {\small$\frac{1}{2^{n+1}}+j\epsilon$} (un);
			\path[->] (un) edge [loop above] node [midway, above] {$1$} (un);
			
			\node[state] (ut) at (10,4) {$s_j$};
			\node[BoxStyle] (tt) at (10,3.2){};
			\node[state] (t2) at (12.5,1.8) {$ t_{n}$};
			\node at (11,1.8) {$\cdots$};
			\node[state] (ti) at (10,1.8) {$t_{i}$};
			\node at (9,1.8) {$\cdots$};
			\node[state] (t1) at (7.5,1.8) {$t_{1}$};
			\node[state, fill=green] (tn) at (12.5,3.2) {$x$};
			\path[-] (ut) edge node  [near end,left] {} (tt);
			\path[->] (tt) edge [out=180,in=180,looseness=3] node  [midway,above,xshift=-1cm] {\small $\frac{1}{2}+(j+1)\epsilon$} (ut);
			\path[->] (tt) edge node  [midway,left, xshift=-0.2cm] {$\frac{1}{4}$} (t1);
			\path[->] (tt) edge node  [midway,left, yshift=-0.1cm] {$\frac{1}{2^{i+1}}$} (ti);	
			\path[->] (tt) edge node  [midway,right, xshift=0.2cm] {$\frac{1}{2^{n+1}}$} (t2);
			\path[->] (tt) edge node  [midway,above] {\tiny $\frac{1}{2^{n+1}}-(j+1)\epsilon$} (tn);
			\path[->] (tn) edge [loop above] node [midway, above] {$1$} (tn);
			\end{tikzpicture}} \caption{$s_j$ and $t_j$ for odd $j \in \{1,\ldots,n\}$}\label{fig:even-n-odd-j}
	\end{subfigure}
	\begin{subfigure}[h]{\textwidth}
		\resizebox{\textwidth}{!}{
			\begin{tikzpicture}
			\node[state] (us) at (3,4) {$s_j$};
			\node[BoxStyle] (st) at (3,3.2){};
			\node[state] (u1) at (0.5,1.8) {$s_{1}$};
			\node at (2,1.8) {$\cdots$};	
			\node[state] (u2) at (3,1.8) {$s_{i}$};
			\node at (4,1.8) {$\cdots$};	
			\node[state] (u3) at (5.5,1.8) {$s_{n}$};	
			\node[state, fill=green] (un) at (5.5,3.2) {$x$};
			\path[-] (us) edge node  [near end,left] {} (st);
			\path[->] (st) edge [out=180,in=180,looseness=3] node  [midway,left] {$\frac{1}{2}-j\epsilon$} (us);
			\path[->] (st) edge node  [midway,left, xshift=-0.2cm] {$\frac{1}{4}$} (u1);
			\path[->] (st) edge node  [midway,left,yshift=-0.1cm] {$\frac{1}{2^{i+1}}$} (u2);
			\path[->] (st) edge node  [midway,right, xshift=0.2cm] {$\frac{1}{2^{n+1}}$} (u3);
			\path[->] (st) edge node  [midway,above] {\small$\frac{1}{2^{n+1}}+j\epsilon$} (un);
			\path[->] (un) edge [loop above] node [midway, above] {$1$} (un);
			
			\node[state] (ut) at (10,4) {$t_j$};
			\node[BoxStyle] (tt) at (10,3.2){};
			\node[state] (t2) at (12.5,1.8) {$ t_{n}$};
			\node at (11,1.8) {$\cdots$};
			\node[state] (ti) at (10,1.8) {$t_{i}$};
			\node at (9,1.8) {$\cdots$};
			\node[state] (t1) at (7.5,1.8) {$t_{1}$};
			\node[state, fill=green] (tn) at (12.5,3.2) {$x$};
			\path[-] (ut) edge node  [near end,left] {} (tt);
			\path[->] (tt) edge [out=180,in=180,looseness=3] node  [midway,above,xshift=-1cm] {\small $\frac{1}{2} +(j+1)\epsilon$} (ut);
			\path[->] (tt) edge node  [midway,left, xshift=-0.2cm] {$\frac{1}{4}$} (t1);
			\path[->] (tt) edge node  [midway,left, yshift=-0.1cm] {$\frac{1}{2^{i+1}}$} (ti);	
			\path[->] (tt) edge node  [midway,right, xshift=0.2cm] {$\frac{1}{2^{n+1}}$} (t2);
			\path[->] (tt) edge node  [midway,above] {\tiny $\frac{1}{2^{n+1}}-(j+1)\epsilon$} (tn);
			\path[->] (tn) edge [loop above] node [midway, above] {$1$} (tn);
			\end{tikzpicture}} \caption{$s_j$ and $t_j$ for even $j \in \{1,\ldots,n\}$}\label{fig:even-n-even-j}
	\end{subfigure}
	\caption{The LMC $\M(2n)$ in which $s \sim_{\epsilon} t$ but not $s R_{2n\epsilon} t$ }\label{fig:example-not-subseteq-R-n-epsilon-even}
\end{figure}

\begin{figure}[H]
	\centering
	\begin{tikzpicture}
	\tikzstyle{BoxStyle} = [draw, circle, fill=black, scale=0.05,minimum width = 0.001pt, minimum height = 0.001pt]
	
	\node[state] (t) at (12,4) {$u$};
	\node[state, fill=green] (y) at (12,2.5) {$x$};
	\path[->] (t) edge [loop left] node  [midway,left] {$1 - \frac{1}{2^{n+1}}$} (t);
	\path[->] (t) edge node  [midway,right,xshift=0.1cm] {$\frac{1}{2^{n+1}}$} (y);
	\path[->] (y) edge [loop right] node [midway, right] {$1$} (y);
	\node at (12, 1.5) {(a) $\epsilon'$-quotient ($\epsilon'$ is at least $2n\epsilon$)};	
	
	\node[state] (t) at (18,4) {$u$};
	\node[state, fill=green] (y) at (18,2.5) {$x$};
	\path[->] (t) edge [loop left] node  [midway,left] {$1 - \frac{1}{2^{n+1}} + \frac{\epsilon}{2}$} (t);
	\path[->] (t) edge node  [midway,right,xshift=0.1cm] {$\frac{1}{2^{n+1}} - \frac{\epsilon}{2}$} (y);
	\path[->] (y) edge [loop right] node [midway, right] {$1$} (y);
	\node at (18, 1.5) {(b) $\epsilon'$-quotient  ($\epsilon'$ is at least $(2n+1)\epsilon$)};	
	\end{tikzpicture}
	\caption{(a) An $\epsilon'$-quotient with $\epsilon'$ at least $2n\epsilon$ obtained by merging the states that are related by the transitive closure of $\sim_{\epsilon}$ comprising $s$, $t$, $s_j$ where $j \in \{1, \ldots, n\}$ and $t_j$ where $j \in \{1, \ldots, n-1\}$; (b) An $\epsilon'$-quotient with $\epsilon'$ at least $(2n+1)\epsilon$ obtained by merging the states that are related by the transitive closure of $\sim_{\epsilon}$ comprising $s$, $t$, $s_j$ and $t_j$ for all $j \in \{1, \ldots, n\}$.}
	\label{fig:exampleMerge2n} 
\end{figure}
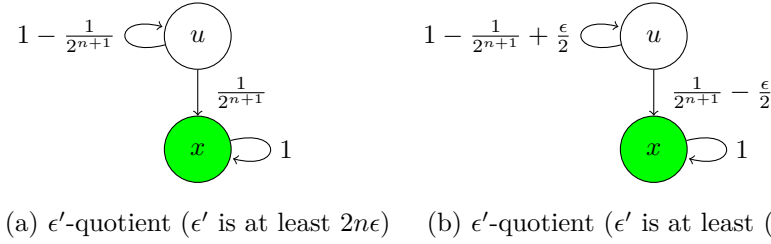

 
\begin{lemma}\label{lemma:adjust-probability-distribution-exists}
	Let $X$ be a partition of $S$. Let $\mu \in\Dist(S)$ and $\gamma \in \Dist(X)$. One can compute in polynomial time a probability distribution $\nu \in \Dist(S)$ such that $(\nu(E))_{E \in X} = \gamma$ and $\|\mu - \nu\|_1 = \|(\mu(E))_{E \in X} - (\nu(E))_{E \in X}\|_1$.
\end{lemma}
\begin{proof}
	\cref{alg:adjust-transition-probability}: We can define such a probability distribution $\nu$ by initializing it to $\mu$. For each $E \in X$, if $\mu(E) \le \gamma(E)$, we choose a state $x \in E$ and increase the value of $\nu(x)$ by $\mu(E) - \gamma(E)$, as shown on line~5 of \cref{alg:adjust-transition-probability}. Otherwise $\mu(E) \gr \gamma(E)$, we run Algorithm~\ref{alg:adjust-transition-probability2} to decrease $\nu(x)$ at one or more states $x \in E$ such that $\nu(E) = \gamma(E)$.
\end{proof}
\noindent
\begin{minipage}[h]{0.53\textwidth}
	\vspace{0pt}  
	\setlength{\algomargin}{0.001in}
	\begin{algorithm}[H]
		\DontPrintSemicolon
		\KwIn{$\mu \in \Dist(S)$, a partition $X$ of $S$ and $\gamma \in \Dist(X)$}
		\KwOut{$\nu \in \Dist(S)$ such that $(\nu(E))_{E \in X} = \gamma$ and $\|\mu - \nu\|_1 = \|(\mu(E))_{E \in X} - (\nu(E))_{E \in X}\|_1$}
		$\nu := \mu$ \;
		\ForEach{$E \in X$}{
			$m: = \mu(E) - \gamma(E)$\\
			\eIf{$m \le 0$}{
				pick $x \in E$; $\nu(x) := \nu(x) + m$\;
			}{
				decreaseProb($\nu$, $E$, $m$)
			}
		}
		\caption{Adjust Probability Distribution}
		\label{alg:adjust-transition-probability}
	\end{algorithm}
\end{minipage}
~
\begin{minipage}[h]{0.45\textwidth}
	\setlength{\algomargin}{0.01in}
	\begin{algorithm}[H] 
		\LinesNotNumbered
		\DontPrintSemicolon
		\KwIn{$\nu \in \Dist(S)$, $E \subseteq S$ and $m$}
		\KwOut{$\nu(E)$ decreased by $m$}
		\ForEach{$x \in E$}{
			\eIf{$\nu(x) \ge m$}{
				$\nu(x) := \nu(x) - m$\;
				$m := 0$ \;	
			}{
				$m := m - \nu(x)$ \;	
				$\nu(x) := 0$\;
			}
			\If{ $m = 0$}{
				break\;
			}
		}
		\caption{decreaseProb}
		\label{alg:adjust-transition-probability2}
	\end{algorithm}
\end{minipage}

\lemmaAdditivityProperty*

\begin{proof}
	Let $\M_1 = <S_1, L, \tau_1, \ell_1>$, $\M_2 = <S_2, L, \tau_2, \ell_2>$ and $\M_3 = <S_3, L, \tau_3, \ell_3>$ be three LMCs. In addition, let $\M_2$ be an $\epsilon_1$-quotient of $\M_1$ and $\M_3$ be an $\epsilon_2$-quotient of $\M_2$. 
	
	By definition, both $\M_2$ and $\M_3$ are quotients, that is, no two states in $\M_2$ or $\M_3$ are probabilistic bisimilar. Let $s_1$ be an arbitrary state from $\M_1$ and $s_2 = [s_1]^{\epsilon_1}_{\M_2}$ be the corresponding state from $\M_2$. The labels of $s_1$ and $s_2$ must coincide. In addition, there exists a transition function $\tau_1': S_1 \times S_1 \to [0, 1]$ such that $\|\tau_1(x_1)- \tau_1'(x_1)\|_1 \le \epsilon_1$ and $x_1 \sim [x_1]^{\epsilon_1}_{\M_2}$ for all $x_1 \in S_1$ in the LMC combining $\M_1' = <S_1, L, \tau_1', \ell_1>$ and $\M_2$. 
	
	Let $s_3 = [s_2]^{\epsilon_2}_{\M_3}$ from $\M_3$ be the state that corresponds to $s_2$. Similarly, the labels of $s_2$ and $s_3$ must be the same. Also, there exists a transition function $\tau_2': S_2 \times S_2 \to [0, 1]$ such that $\|\tau_2(x_2) -\tau_2'(x_2)\|_1\le \epsilon_2$ and $x_2 \sim [x_2]^{\epsilon_2}_{\M_3}$ for all $x_2 \in S_2$ in the LMC combining $\M_2' = <S_2, L, \tau_2', \ell_2>$ and $\M_3$.
	
	To prove $\M_3$ is an $(\epsilon_1 + \epsilon_2)$-quotient of $\M_1$, we show the existence of a transition function $\tau_1'': S_1 \to \Dist(S_1)$ such that $\|\tau_1(x_1) -\tau_1''(x_1)\|_1\le \epsilon_1+ \epsilon_2$ and $x_1 \sim [x_2]^{\epsilon_2}_{\M_3}$ for all $x_1 \in S_1$ in the LMC combining $\M_1'' = <S_1, L, \tau_1'', \ell_1>$ and $\M_3$ where $x_2 = [x_1]^{\epsilon_1}_{\M_2}$.
	
	
	Let us define a function $f: S_2 \to 2^{S_1}$ such that $f(x_2) = \{x_1 \suchthat x_2 = [x_1]^{\epsilon_1}_{\M_2}\}$. Similarly, define a function $g: S_3 \to 2^{S_2}$ such that $g(x_3) = \{x_2 \suchthat x_3 = [x_2]^{\epsilon_2}_{\M_3}\}$. The function $f$ induces a partition $X_1$ over $S_1$: $X_1 = \{ f(x_2) \suchthat x_2 \in S_2\}$. 
	
	Let $x_1 \in S_1$ and $x_2 = [x_1]^{\epsilon_1}_{\M_2}$. We have
	\begin{align}
	&\tau_1'(x_1)(f(y_2))= \tau_2(x_2)(y_2) \text{ for all $y_2 \in S_2$, that is,  }(\tau_1'(x_1)(E))_{E \in X_1} = \tau_2(x_2) \label{eqn:property-of-X1}.
	\end{align} 
	
	Similarly, 	let $x_2 \in S_2$ and $x_3 = [x_2]^{\epsilon_2}_{\M_3}$. We have
	\begin{align}
	&\tau_2'(x_2)(g(y_3))= \tau_3(x_3)(y_3) \text{ for all $y_3 \in S_3$}\label{eqn:property-of-g}.
	\end{align} 
	
	By \cref{lemma:adjust-probability-distribution-exists}, given $X_1$, $\tau_1'(x_1)$ and $\tau_2'(x_2)$, one can compute in polynomial time a probability distribution $\nu_{x_1}$ such that 
	\begin{align}
	&\nu_{x_1}(f(y_2))= \tau_2'(x_2)(y_2) \text{ for all } y_2 \in S_2 \text{, that is,  }  \label{eqn:property1-of-nu-s1}\\
	&(\nu_{x_1}(E))_{E \in X_1} = \tau_2'(x_2) \text{ and }  \nonumber\\
	&\|\nu_{x_1} -  \tau_1'(x_1)\|_1 = \| (\nu_{x_1}(E))_{E \in X_1} - (\tau_1'(x_1)(E))_{E \in X_1}\|_1 \label{eqn:property2-of-nu-s1}.
	\end{align}
	
	For each $x_1 \in S_1$, we compute such $\nu_{x_1}$. Define $\tau_1''$ as $\tau_1''(x_1) = \nu_{x_1}$ for all $x_1 \in S_1$. We first show that $\|\tau_1''(x_1)- \tau_1(x_1)\|_1 \le \epsilon_1+\epsilon_2$ for all $x_1 \in S_1$. Let $x_1 \in S_1$ and $x_2 = [x_1]^{\epsilon_1}_{\M_2}$. We have 
	\begin{eqnarray*} 
		&&    \|\tau_1''(x_1) - \tau_1(x_1)\|_1 \\
		&=& \|\nu_{x_1} - \tau_1(x_1)\|_1 \\
		&\le& \|\nu_{x_1} - \tau_1'(x_1)\|_1 +   \|\tau_1'(x_1) - \tau_1(x_1)\|_1 \commenteq{triangle inequality}\\
		& \le & \|\nu_{x_1} - \tau_1'(x_1)\|_1 +   \epsilon_1 \commenteq{by assumption}\\
		& = & \sum_{E \in X_1} |\nu_{x_1}(E) - \tau_1'(x_1)(E)| +   \epsilon_1\commenteq{\cref{eqn:property2-of-nu-s1}}\\
		& = & \sum_{y_2 \in S_2} |\nu_{x_1}(f(y_2)) - \tau_1'(x_1)(f(y_2))| +   \epsilon_1\commenteq{definition of $X_1$}\\ 	 	
		& = &  \sum_{y_2 \in S_2} |\tau_2'(x_2)(y_2) - \tau_2(x_2)(y_2)| +   \epsilon_1\commenteq{\cref{eqn:property1-of-nu-s1} and \cref{eqn:property-of-X1}}\\ 	 	
		& \le &  \epsilon_2 +   \epsilon_1\commenteq{by assumption}\\ 	 	 			 			
	\end{eqnarray*}
	
	Let $\M_1'' = <S_1, L, \tau_1'', \ell_1>$. It remains to show that $x_1 \sim [x_2]^{\epsilon_2}_{\M_3}$ for all $x_1 \in S_1$ in the LMC combining $\M_1'' = <S_1, L, \tau_1'', \ell_1>$ and $\M_3$ where $x_2 = [x_1]^{\epsilon_1}_{\M_2}$. 
	
	
	Define a function $h: S_3 \to 2^{S_1}$ such that $h(x_3) = \{x_1 \suchthat x_2 \in g(x_3) \land x_1 \in f(x_2) \}$. A partition $X_1'$ over $S_1$ is induced by $h$: $X_1' = \{ h(x_3) \suchthat x_3 \in S_3 \}$. It suffices to show that $x_3 \sim x_1$ for any $x_3 \in S_3$ and any $x_1 \in h(x_3)$ in the LMC combining $\M_1''$ and $\M_3$. Let $x_1 \in S_1$,  $x_2 = [x_1]^{\epsilon_1}_{\M_2}$ and $x_3 = [x_2]^{\epsilon_2}_{\M_3}$. We have $x_1 \in h(x_3)$. Let $y_3 \in S_3$. Then,
	\begin{eqnarray*} 
		&&   \tau_1''(x_1)(h(y_3))  \\
		&=& \nu_{x_1}(h(y_3))  \commenteq{definition of $\tau_1''$}\\
		&=& \sum_{y_1 \in h(y_3)} \nu_{x_1}(y_1)  \\
		&=& \sum_{y_2 \in g(y_3)} \sum_{y_1 \in  f(y_2)} \nu_{x_1}(y_1) \commenteq{definition of $h$}\\
		&=& \sum_{y_2 \in g(y_3)} \nu_{x_1}(f(y_2)) \\
		&=& \sum_{y_2 \in g(y_3)} \tau_2'(x_2)(y_2)\commenteq{\cref{eqn:property1-of-nu-s1}}\\ 
		&=& \tau_2'(x_2)(g(y_3))\\ 
		&=& \tau_3(x_3)(y_3)\commenteq{\cref{eqn:property-of-g}} 				
	\end{eqnarray*}
	
	Since $\tau_1''(x_1)(h(y_3)) = \tau_3(x_3)(y_3) $ for all $y_3 \in S_3$, we have $x_1 \sim x_3$ in the LMC combining $\M_1''$ and $\M_3$. By definition of approximate quotient, we have that $\M_3$ is an $(\epsilon_1 + \epsilon_2)$-quotient of $\M_1$.
\end{proof}

\newpage
\section{Proofs of \cref{subsection:local-bisimilarity-distances}}\label{appendix:local-bisimilarity-distance}

\begin{lemma}\label{lemma:local-bisimilarity-distances-witness-transition-funciton}
Let us define a transition function $\tauHyp' \in {\rm T}$ as 
\[
	\tauHyp'(x) = \left \{
	\begin{array}{ll}
	\tauHyp(x)& \mbox{if $x \not\in \{s, t\}$}\\
	\nu_x & \mbox{otherwise}
	\end{array}
	\right .
\]
where $\nu_s$ (resp. $\nu_t$) is computed by running \cref{alg:adjust-transition-probability} with $X$, $\tauHyp(s)$ (resp. $\tauHyp(t)$) and $\frac{(\tauHyp(s)(E))_{E \in X} + (\tauHyp(t)(E))_{E \in X}}{2}$. We have $d_{\local}^{\Hyp}(s, t) = \max\{\|\tauHyp'(s) - \tauHyp(s)\|_1, \|\tauHyp'(t) - \tauHyp(t)\|_1 \} $. 
\end{lemma}
\begin{proof}
	Let $\gamma = \frac{(\tauHyp(s)(E))_{E \in X} + (\tauHyp(t)(E))_{E \in X}}{2}$. The following function $f: \Dist(X) \to \mathbb{R}$ attains its minimum at $\gamma$: $f(\phi) = \max\{\|\phi - (\tauHyp(s)(E))_{E \in X}\|_1, \|\phi - (\tauHyp(t)(E))_{E \in X}\|_1\} \text{ where } \phi \in \Dist(X)$.

By running \cref{alg:adjust-transition-probability} with $X$, $\tauHyp(s)$ (resp. $\tauHyp(t)$) and $\gamma$, we compute a probability distribution $\nu_s$ (resp. $\nu_t$). By \cref{lemma:adjust-probability-distribution-exists}, we have that $(\nu_s(E))_{E \in X} = (\nu_t(E))_{E \in X}  = \gamma$ and $\|\nu_s - \tauHyp(s)\|_1 = \|\gamma - (\tauHyp(s)(E))_{E \in  X}\|_1 = \|\gamma - (\tauHyp(t)(E))_{E \in  X}\|_1 = \|\nu_t - \tauHyp(t)\|_1$. Define the probability transition function $\tauHyp' \in {\rm T}$ as 	
\[
\tauHyp'(x) = \left \{
\begin{array}{ll}
\tauHyp(x)& \mbox{if $x \not\in \{s, t\}$}\\
\nu_x & \mbox{otherwise}
\end{array}
\right .
\]

Let $\Hyp' = <S, L, \tauHyp', \ell>$. We have $s\sim_{\Hyp'} t$ since $ (\tauHyp'(s)(E))_{E \in X} =(\nu_s(E))_{E \in X} = \gamma = (\nu_t(E))_{E \in X} = (\tauHyp'(t)(E))_{E \in  X}$.  It remains to show that $\tauHyp'= \arg\min_{\tau \in {\rm T} } \max\{ \|\tauHyp'(s) - \tauHyp(s)\|_1 , \|\tauHyp'(t) - \tauHyp(t)\|_1\}$. Let $\tau \in {\rm T}$. Then,
\begin{eqnarray*} 
	&&\max\big\{\|\tau(s) -\tauHyp(s)\|_1, \|\tau(t) -\tauHyp(t)\|_1\big\} \\
	&=& \max\big\{ \sum_{v \in S} |\tau(s)(v) -\tauHyp(s)(v)|, \sum_{v \in S} |\tau(t)(v)-\tauHyp(t)(v)|\big\}\\
	&=&\max\big\{\sum_{E \in X } \sum_{v \in E} |\tau(s)(v) -\tauHyp(s)(v)|, \\
	&& \phantom{\max;;}\sum_{E \in X } \sum_{v \in E} |\tau(t)(v)-\tauHyp(t)(v)|\big\}\\
	&\ge& \max\big\{\|(\tau(s)(E))_{E \in X} -(\tauHyp(s)(E))_{E \in X}\|_1 ,  \\
	&& \phantom{\max;;} \|(\tau(t)(E))_{E \in X} -(\tauHyp(t)(E))_{E \in X}\|_1\big\}\commenteq{triangle inequality}.\\
	&\ge& \max\big\{\|\gamma -(\tauHyp(s)(E))_{E \in X}\|_1 ,  \|\gamma-(\tauHyp(t)(E))_{E \in X}\|_1\big\}\\
	&& \phantom{\max;;} \commenteq{$(\tau(s)(E))_{E \in X}  = (\tau(t)(E))_{E \in X} $ since $s \mathord{\equiv_{X}} t$; $f$ attains minimum at $\gamma$}\\
	&=& \max\big\{\|(\tauHyp'(s)(E))_{E \in X} -(\tauHyp(s)(E))_{E \in X}\|_1, \\
	&& \phantom{\max;;} \|(\tauHyp'(t)(E))_{E \in X} -(\tauHyp(t)(E))_{E \in X}\|_1\big\}\\
	&=& \max\big\{\|\tauHyp'(s) -\tauHyp(s)\|_1, \|\tauHyp'(t) -\tauHyp(t)\|_1\big\} 
\end{eqnarray*}	
\end{proof}

\propositionAdjustTransitionFunction*

\begin{proof}
Let $\gamma = \frac{(\tauHyp(s)(E))_{E \in X} + (\tauHyp(t)(E))_{E \in X}}{2}$. By running \cref{alg:adjust-transition-probability} with $X$, $\tauHyp(s)$ (resp. $\tauHyp(t)$) and $\gamma$, we compute a probability distribution $\nu_s$ (resp. $\nu_t$). By \cref{lemma:adjust-probability-distribution-exists}, we have that $(\nu_s(E))_{E \in X} = (\nu_t(E))_{E \in X}  = \gamma$ and $\|\nu_s - \tauHyp(s)\|_1 = \|\gamma - (\tauHyp(s)(E))_{E \in  X}\|_1 = \|\gamma - (\tauHyp(t)(E))_{E \in  X}\|_1 = \|\nu_t - \tauHyp(t)\|_1$. Define the probability transition function $\tauHyp' \in {\rm T}$ as 	
\[
\tauHyp'(x) = \left \{
\begin{array}{ll}
\tauHyp(x)& \mbox{if $x \not\in \{s, t\}$}\\
\nu_x & \mbox{otherwise.}
\end{array}
\right .
\]

It follows from \cref{lemma:local-bisimilarity-distances-witness-transition-funciton} that $d_{\local}^{\Hyp}(s, t) = \max\{\|\tauHyp'(s) - \tauHyp(s)\|_1, \|\tauHyp'(t) - \tauHyp(t)\|_1 \} $. Then, $d_{\local}^{\Hyp}(s, t) = \max\{\|\gamma - (\tauHyp(s)(E))_{E \in X}\|_1, \|\gamma - (\tauHyp(t)(E))_{E \in X}\|_1\} = \frac{1}{2} \|(\tauHyp(s)(E))_{E \in X} - (\tauHyp(t)(E))_{E \in X}\|_1$.
\end{proof}

The following lemma holds for the LMCs $Q_i$ in \cref{alg:local-distance-merge-algorithm}.
\begin{restatable}{lemma}{lemmaSmallGlobalDistanceLoopInvariant}\label{lemma:small-global-distance-loop}
	For all $i \in \nat$, we have $Q_{i+1}$ is an $\epsilon_2$-quotient of $Q_{i}$.
\end{restatable}
\begin{proof}
	Let $i \in \nat$. Assume the algorithm steps into the $i$'th iteration and selects state $s$ and $t$ from $\Q_i$ which have the least local bisimilarity distance. 
	
	Let $\gamma = \frac{(\tau^{\Q_i}(s)(E))_{E \in X_i} + (\tau^{\Q_i}(t)(E))_{E \in X_i}}{2}$. By running \cref{alg:adjust-transition-probability} with $X$, $\tau^{\Q_i}(s)$ (resp. $\tau^{\Q_i}(t)$) and $\gamma$, we compute a probability distribution $\nu_s$ (resp. $\nu_t$). By \cref{lemma:adjust-probability-distribution-exists}, we have that $(\nu_s(E))_{E \in X_i} = (\nu_t(E))_{E \in X_i}  = \gamma$ and $\|\nu_s - \tau^{\Q_i}(s)\|_1 = \|\gamma - (\tau^{\Q_i}(s)(E))_{E \in  X_i}\|_1 = \|\gamma - (\tau^{\Q_i}(t)(E))_{E \in  X_i}\|_1 = \|\nu_t - \tau^{\Q_i}(t)\|_1$. Define the probability transition function $\tau' \in {\rm T}$ as 	
	\[
	\tau'(x) = \left \{
	\begin{array}{ll}
	\tau^{\Q_i}(x)& \mbox{if $x \not\in \{s, t\}$}\\
	\nu_x & \mbox{otherwise.}
	\end{array}
	\right .
	\]
	
	It follows from \cref{lemma:local-bisimilarity-distances-witness-transition-funciton} that $d_{\local}^{\Q_{i}} (s, t) = \max\{\|\tau^{\Q_{i}}(s) - \tau'(s) \|_1, \|\tau^{\Q_{i}}(t) - \tau'(t)\|_1\}$. By the definition of $\tau'$, for all $u$ from $\Q_{i}$ we have \begin{equation}\label{eqn:bound-L1-distance-for-all-states}
	\|\tau^{\Q_{i}}(u) - \tau'(u)\|_1 \le \epsilon_2
	\end{equation} 
	since $\|\tau^{\Q_{i}}(u) - \tau'(u)\|_1 \le \max\{\|\tau^{\Q_{i}}(s) - \tau'(s) \|_1, \|\tau^{\Q_{i}}(t) - \tau'(t)\|_1\} = d_{\local}^{\Q_{i}} (s, t) \le \epsilon_2$.

	It is not hard to prove that the LMC $\Hyp' = <S^{\Q_i}, L, \tau', \ell^{\Q_{i}}>$ is probabilistic bisimilar to the LMC $\M_{i+1}'$. Then, since $\Q_{i+1}$ is probabilistic bisimilar to $\M_{i+1}'$, we have $\Q_{i+1}$ is probabilistic bisimilar to $\Hyp'$. Thus, by \eqref{eqn:bound-L1-distance-for-all-states} and that $Q_{i+1}$ is a quotient, we have $Q_{i+1}$ is an $\epsilon_2$-quotient of $Q_i$.

\end{proof}

%

\section{Proofs of \cref{subsection:approximate-partition-refinement}}\label{appendix: approximate-partition-refinement}

The following lemma holds for the LMCs $Q_i$ in \cref{alg:approximate-partition-refinement-merge-algorithm}.
\begin{lemma}\label{lemma:small-global-distance-loop-approximate-partition-refinement}
	For all $i \in \nat$, we have $Q_{i+1}$ is an $\epsilon_2$-quotient of $Q_{i}$.
\end{lemma}

\begin{proof}
	Let $i \in \nat$. Let $x \in S^{\Q_i}$ and $E_x \in  X_i$ be the set such that $x \in E_x$. Let $\gamma_{E_x} = \sum_{u \in  E_x} \frac{(\tau^{\Q_i}(u)(E))_{E \in X}}{|E_x|}$ be a probability distribution over $X_i$. We run \cref{alg:adjust-transition-probability} with $X, \tau^{\Q_i}(x), \gamma_{E_x}$ to obtain $\nu_x$, a probability distribution over $S^{\Q_i}$. By \cref{lemma:adjust-probability-distribution-exists}, we have $(\nu_x(E))_{E \in X_i} = \gamma_{E_x}$ and $\|\nu_x - \tau^{\Q_i}(x)\|_1 = \sum_{E \in X_i} |\nu_x(E) - \tau^{\Q_i}(x)(E)|$. Define the transition function $\tau'$ as $\tau'(x) = \nu_x$ for all $x \in S^{\Q_i}$.
	
	We have  
	\begin{equation}\label{eqn:bound-transition-function-2}
	\|\tau'(x) - \tau^{\Q_i}(x)\|_1 \le \epsilon_2
	\end{equation}
	since 
	\begin{eqnarray*}
		&& |\tau'(x) - \tau^{\Q_i}(x)| \\
		&=& \sum_{E \in X_i}  |\nu_x(E) - \tau^{\Q_i}(x)(E)| \\
		& = & \sum_{E \in X_i} |\gamma_{E_x}(E) - \tau^{\Q_i}(x)(E)|\\
		& = & \sum_{E \in X_i} |\big(\sum_{u \in  E_x} \frac{\tau^{\Q_i}(u)(E)}{|E_x|}\big) - \tau^{\Q_i}(x)(E)|\\
		& = & \sum_{E \in X_i}|\sum_{u \in  E_x} \frac{\tau^{\Q_i}(u)(E) - \tau^{\Q_i}(x)(E)}{|E_x|} | \\
		& \le & 	\sum_{E \in X_i}| \sum_{u \in  E_x} \frac{|\tau^{\Q_i}(u)(E) - \tau^{\Q_i}(x)(E)|}{|E_x|} | \commenteq{triangle inequality}\\
		& = & 	\sum_{u \in  E_x}  \sum_{E \in X_i} \frac{|\tau^{\Q_i}(u)(E) - \tau^{\Q_i}(x)(E)|}{|E_x|} \\
		& = & 	\sum_{u \in  E_x}  \frac{\sum_{E \in X_i} |\tau^{\Q_i}(u)(E) - \tau^{\Q_i}(x)(E)|}{|E_x|} \\
		&\le& \sum_{u \in  E_x} \frac{\epsilon_2}{|E_x|}  \commenteq{$\forall u \in E_x: \sum_{E \in X_i}|\tau^{\Q_i}(u)(E) - \tau^{\Q_i}(x)(E)| \le \epsilon_2 $}\\
		& =& \epsilon_2.
	\end{eqnarray*}

	Let us consider the LMC $\Q_i' = <S^{Q_i}, L, \tau', \ell^{\Q_{i}}>$. Since for all $s, t$ in the same set $E$ of $X_i$ we have $\ell^{\Q_{i}}(s) = \ell^{\Q_{i}}(t)$ and $\tau'(s) = \tau'(t)$, states in the same set of $E$ of $X_i$ are probabilistic bisimilar in the LMC $\Q_i'$. It is not hard to prove that $\Q_i'$ is probabilistic bisimilar to $\M_{i+1}$. Since $\Q_{i+1}$ is probabilistic bisimilar to $\M_{i+1}$, $\Q_{i+1}$ is probabilistic bisimilar to $\Q_{i}'$. By \eqref{eqn:bound-transition-function-2} and that $\Q_{i+1}$ is a quotient LMC, we have $\Q_{i+1}$ is an $\epsilon_2$-quotient of $\Q_i$.
\end{proof}

\theoremBoundGlobalDistanceApproximatePartitionRefinement*
\begin{proof}
	Let $i \in \nat$. The first part of the statement follows from \cref{lemma:small-global-distance-loop} for $\Q_i$ in \cref{alg:local-distance-merge-algorithm}. The first part of the statement follows from \cref{lemma:small-global-distance-loop-approximate-partition-refinement} for $\Q_i$ in \cref{alg:approximate-partition-refinement-merge-algorithm}. 
	
	We prove the second part of the statement by induction. The base case $i = 0$ follows from the fact that $\Q_0 = \Hyp /_{\sim_{\Hyp}}$. For the induction step, we assume $\Q_{i}$ is an $i\epsilon_2$-quotient of $\Hyp$. Then, from the first part that $\Q_{i+1}$ is an $\epsilon_2$-quotient of $\Q_{i}$ and the additivity lemma, we have $\Q_{i+1}$ is an $(i+1)\epsilon_2$-quotient of $\Hyp$.
\end{proof}

\corollaryBoundQuotientError*
\begin{proof}
	It is not hard to see that $\Q_0$ is an $\epsilon$-quotient of $\M$. It then follows from \cref{theorem:bounding-global-distance} and \cref{lemma:additivity-property} that $\Q_i$ is an $(\epsilon+i\epsilon_2)$-quotient of $\M$.
\end{proof}

\section{More Experimental Results}\label{appendix:more-results}

Let us fix the error bound $\delta = 0.01$. In the tables, local and apr stand for the minimisation algorithms using local bisimilarity distance and approximate partition refinement, respectively. The rows are highlighted in yellow when the structure of the quotient of the original model is successfully recovered. While the rows are highlighted in red when $\epsilon_2$ is too big, that is, \modify{the minimisation algorithms aggressively merge some states in the perturbed LMCs, resulting in quotients of which the size are smaller than that of the quotients of the unperturbed LMCs}. 


\begin{figure}[h!]
\begin{floatrow}
	\capbtabbox[.49\textwidth]{
		\noindent\begin{tabular}{|c|c|c|c|}
			\hline 
			\multicolumn{4}{|c|}{Herman3}\\
			\hline
			\multirow{1}{*}{\shortstack[l]{$\epsilon = 0.001$}}& 
			\# states&	\# trans&	\# iter\\
			\hline 						
			$\M$ \& $\Hyp$	&		 		8	  &			28				&\\
			$\M/_{\sim_{\M}}$	&		 		2	  &			3				&\\		
			$\Hyp/_{\sim_{\Hyp}}$      & 		3  &	       7             & \\
			\hline
			\multicolumn{4}{|c|}{Perturbed LMC \#1 - \#5}\\ 
			\hline
			\multicolumn{4}{|c|}{$\epsilon_2 \in\{ 0.00001,  0.0001\}$}\\
			\hline
			local      & 3 & 7 & 0  \\
			apr		  & 3 & 7 & 0  \\
			\hline
			\multicolumn{4}{|c|}{$\epsilon_2 \in \{0.001, 0.01, 0.1\}$}\\
			\hline
			\rowcolor{yellow}
			local     & 2 & 3 & 1  \\
			\rowcolor{yellow}	
			apr		  & 2 & 3 & 1  \\
			\hline	
			\rowcolor{white}
			\multicolumn{4}{c}{}\\
			\multicolumn{4}{c}{}\\
			\multicolumn{4}{c}{}\\
			\multicolumn{4}{c}{}\\
			\multicolumn{4}{c}{}\\
		\end{tabular}
	}{}
	\capbtabbox[.49\textwidth]{
		\begin{tabular}{|c|c|c|c|}
		\hline 
		\multicolumn{4}{|c|}{Herman7}\\
		\hline
		\multirow{1}{*}{\shortstack[l]{$\epsilon = 0.001$}}& 
		\# states&	\# trans&	\# iter\\
		\hline 						
		$\M$ \& $\Hyp$	&		 		128	  &			2188				&\\
		$\M/_{\sim_{\M}}$	&		 		9	  &			49				&\\		
		$\Hyp/_{\sim_{\Hyp}}$      & 		115  &	       1925            & \\
		\hline
		\multicolumn{4}{|c|}{Perturbed LMC \#2}\\ 
		\hline
		\multicolumn{4}{|c|}{$\epsilon_2 \in\{ 0.00001,  0.0001\}$}\\
		\hline
		local \& apr      & 115 & 1925 & 0  \\
		\hline
		\multicolumn{4}{|c|}{$\epsilon_2 = 0.001$}\\
		\hline
		local      & 114 & 1809 & 1  \\
		apr		  & 115 & 1925 & 0  \\
		\hline
		\multicolumn{4}{|c|}{$\epsilon_2 = 0.01$}\\
		\hline
		local      & 114 & 1809 & 1  \\
		\rowcolor{yellow}
		apr		  & 9 & 49 & 1  \\
		\rowcolor{white}
		\hline		
		\multicolumn{4}{|c|}{$\epsilon_2 = 0.1$}\\
		\hline
		local      & 114 & 1809 & 1  \\
		\rowcolor{red!40}
		apr		  & 10 & 60 & 1  \\
		\hline
	\end{tabular}	
	}{}
\end{floatrow}
\end{figure}
The next tables show the results of running the two minimisation algorithms on the LMCs that model Herman's self-stabilisation algorithm with $3$ processes and $7$ processes, respectively. All the perturbed LMCs are obtained by sampling with $\epsilon = 0.001$. For the model with $7$ processes, only the minimisation algorithm using approximate partition refinement with $\epsilon_2 = 0.01$ can recover the structure of the quotient of the original model. For this model,  though the minimisation algorithm with $\epsilon_2 = 0.1$ does not perfectly recover the structure of the quotient of the original LMC, the final minimised LMC is quite close. 


\begin{figure}[h!]	
	\begin{floatrow}
		\capbtabbox[.49\textwidth]{
			\noindent\begin{tabular}{|c|c|c|c|}
				\hline 
				\multicolumn{4}{|c|}{Herman13}\\
				\hline
				\multirow{1}{*}{\shortstack[l]{$\epsilon = 0.001$}}& 
				\# states&	\# trans&	\# iter\\
				\hline 						
				$\M$ \& $\Hyp$	&		 8192	& 1594324			&\\
				$\M/_{\sim_{\M}}$	&		 		190	& 12857			&\\		
				$\Hyp/_{\sim_{\Hyp}}$      & 		8167 &	1585929            & \\
				\hline
				\multicolumn{4}{|c|}{Perturbed LMC \#2}\\ 
				\hline
				\multicolumn{4}{|c|}{$\epsilon_2 \in\{ 0.00001, 0.0001\}$}\\
				\hline
				apr      & 8167	&1585929&	0 \\
				\hline
				\multicolumn{4}{|c|}{$\epsilon_2 = 0.001$}\\
				\hline
				apr		  &8166	&1577761&	1 \\
				\hline		
				\multicolumn{4}{|c|}{$\epsilon_2 = 0.01$}\\
				\hline
				apr		  &192	&13250	&1 \\
				\hline							
				\multicolumn{4}{|c|}{$\epsilon_2 = 0.1$}\\
				\hline
				\rowcolor{red!40}	
				apr		   & 7608 &	1439629 &	1  \\
				\hline
							\rowcolor{white}	
				\multicolumn{4}{c}{}\\
				\multicolumn{4}{c}{}\\
			\end{tabular}
		}{}
		\capbtabbox[.49\textwidth]{
			\begin{tabular}{|c|c|c|c|}
				\hline
				\multicolumn{4}{|c|}{Herman15}\\
				\hline
				\multirow{1}{*}{\shortstack[l]{$\epsilon = 0.0001$}}& 
				\# states&	\# trans&	\# iter\\
				\hline 
				$\M$ \& $\Hyp$ &	32768 &	14348908    & \\
				$\M/_{\sim_{\M}}$	&		 		612	& 104721	  		&\\								
				$\Hyp/_{\sim_{\Hyp}}$     & 	32739 &	14323591    & \\
				\hline
				\multicolumn{4}{|c|}{Perturbed LMC \#3}\\ 
				\hline
				\multicolumn{4}{|c|}{$\epsilon_2 = 0.00001$}\\
				\hline
				apr		&  32739 &	14323591 &	0\\
				\hline
				\multicolumn{4}{|c|}{$\epsilon_2 = 0.0001$}\\
				\hline
				apr		&  32738 &	14290851&	1\\
				\hline			
				\multicolumn{4}{|c|}{$\epsilon_2 =  0.001$}\\
				\hline
				apr		 &  3489 &	1918364	& 1 \\
				\hline		
				\multicolumn{4}{|c|}{$\epsilon_2 =  0.01$}\\
				\hline
				\rowcolor{yellow}
				apr		 &  612	& 104721 &	1\\
				\rowcolor{white}
				\hline					
				\multicolumn{4}{|c|}{$\epsilon_2 = 0.1$}\\
				\hline
				\rowcolor{red!40}
				apr		 &  25893	& 11090774	& 1\\
				\hline
			\end{tabular}	
		}{}
	\end{floatrow}
\end{figure}


The table above (left) shows the results of running the minimisation algorithm using approximate partition refinement on the LMC that models Herman's self-stabilisation algorithm with $13$ processes. The perturbed LMCs are obtained by perturbing the probabilities with $\epsilon = 0.001$. For this model, the minimisation algorithm could not recover the structure of the quotient of the original model, but the minimised LMC obtained with $\epsilon_2 = 0.01$ is quite close. The table above (right) shows the results of running the minimisation algorithms using approximate partition refinement on the LMC that models Herman's self-stabilisation algorithm with $15$ processes. The perturbed LMCs are obtained by perturbing the probabilities with $\epsilon = 0.0001$. For this model, the minimisation algorithm successfully recovers the quotient of the original model. However, when $\epsilon_2 = 0.1$, the value is too big that the minimisation algorithm\modify{aggressively merges states in the perturbed model and results in a quotient of which the size is even smaller than that of the quotient of the unperturbed LMC.}

For the LMCs that model the synchronous leader election protocol by Itai and Rodeh, the exact partition refinement can always recover the structure of the quotient of the original model. Let $\epsilon \in \{0.00001 ,0.0001, 0.001, 0.01, 0.1\}$. The tables on the top of the page show the results of running the two minimisation algorithms on the LMCs that model the synchronous leader election protocol by Itai and Rodeh with $N =5, K=5$ and $N =6, K=4$ on the left and right, respectively. All the perturbed LMCs are obtained by perturbing the probabilities. 

\begin{figure}[t]	
	\begin{floatrow}
	\capbtabbox[.49\textwidth]{
	\noindent\begin{tabular}{|c|c|c|c|}
	\hline
	Leader5-5&	\# states&	\# trans&	\# iter\\
	\hline \hline						
	$\M$ \& $\Hyp$		& 	12709   &			15833	    &\\
	$\M /_{\sim_{\M}}$  &			12  &	       13             & \\
	\rowcolor{yellow}
	$\Hyp /_{\sim_{\Hyp}}$  &			12  &	       13             & \\
	\rowcolor{white}
	\hline
	\multicolumn{4}{|c|}{Perturbed LMC \#1-\#5} \\
	\hline
	\multicolumn{4}{|c|}{$\epsilon_2 \in \{0.00001, 0.0001, 0.001, 0.01, 0.1\}$}\\
	\hline
	\rowcolor{yellow}
	 apr \& local	 & 12 & 13 & 0  \\
	\hline
	\end{tabular}
	}{}
\capbtabbox[.49\textwidth]{
	\begin{tabular}{|c|c|c|c|}
		\hline
		Leader6-4&	\# states&	\# trans&	\# iter\\
		\hline \hline						
		$\M$ \& $\Hyp$		& 	20884   &			24979	    &\\
		$\M /_{\sim_{\M}}$  &			14  &	       15             & \\
		\rowcolor{yellow}
		$\Hyp /_{\sim_{\Hyp}}$  &			14  &	       15             & \\
		\rowcolor{white}
		\hline
		\multicolumn{4}{|c|}{Perturbed LMC \#1-\#5} \\
		\hline
		\multicolumn{4}{|c|}{$\epsilon_2 \in \{0.00001, 0.0001, 0.001, 0.01, 0.1\}$}\\
		\hline
		\rowcolor{yellow}
		apr	\& local & 14 & 15 & 0  \\
		\hline
	\end{tabular}
	}{}
\end{floatrow}
\end{figure}


\begin{figure}[h!]
	\begin{floatrow}
		\capbtabbox[.49\textwidth]{
			\noindent\begin{tabular}{|c|c|c|c|}
				\hline
				\multicolumn{4}{|c|}{BRP16-3}\\
				\hline
				\multirow{1}{*}{\shortstack[l]{$\epsilon = 0.01$}}& 
				\# states&	\# trans&	\# iter\\
				\hline 						
				$\M$ \& $\Hyp$	&		 		886	  &			1155				&\\
				$\M/_{\sim_{\M}}$	&		 		440	  &			616				&\\		
				$\Hyp/_{\sim_{\Hyp}}$      & 		671  &	       940             & \\
				\hline
				\multicolumn{4}{|c|}{Perturbed LMC \#3}\\ 
				\hline
				\multicolumn{4}{|c|}{$\epsilon_2 = 0.00001$}\\
				\hline
				apr      & 668 & 936 & 1  \\
				\hline
				\multicolumn{4}{|c|}{$\epsilon_2 = 0.0001$}\\
				\hline
				apr		  & 651 & 912 & 1  \\
				\hline
				\multicolumn{4}{|c|}{$\epsilon_2 = 0.001$}\\
				\hline
				apr		  & 574 & 806 & 2  \\
				\hline	
				\multicolumn{4}{|c|}{$\epsilon_2 = 0.01$}\\
				\hline
				apr		  & 472 & 663 & 2  \\
				\hline						
				\multicolumn{4}{|c|}{$\epsilon_2 = 0.1$}\\
				\hline
				\rowcolor{red!40}
				apr		   & 100 & 195 & 1   \\
				\hline
			\end{tabular}
		}{}
		\capbtabbox[.49\textwidth]{
			\begin{tabular}{|c|c|c|c|}
				\hline
				\multicolumn{4}{|c|}{BRP64-4}\\
				\hline
				\multirow{1}{*}{\shortstack[l]{$\epsilon = 0.001$}}& 					
				\# states&	\# trans&	\# iter\\
				\hline 
				$\M$ \& $\Hyp$ &	4359		 &			5763    & \\
				$\M/_{\sim_{\M}}$	&		 		2185	&  3081	  		&\\								
				$\Hyp/_{\sim_{\Hyp}}$     & 		3453	& 4857      & \\
				\hline
				\multicolumn{4}{|c|}{Perturbed LMC \#4}\\ 
				\hline
				\multicolumn{4}{|c|}{$\epsilon_2 = 0.00001$}\\
				\hline
				apr		&  3377 &	4754 &	1\\
				\hline
				\multicolumn{4}{|c|}{$\epsilon_2 = 0.0001$}\\
				\hline
				apr		&  3034 &	4279&	2\\
				\hline			
				\multicolumn{4}{|c|}{$\epsilon_2 =  0.001$}\\
				\hline
				apr		 &  2350 &	3318&	4 \\
				\hline
				\multicolumn{4}{|c|}{$\epsilon_2 =  0.01$}\\
				\hline
				\rowcolor{yellow}
				apr		 &  2185 &	3081 &	1 \\
				\rowcolor{white}
				\hline			
				\multicolumn{4}{|c|}{$\epsilon_2 = 0.1$}\\
				\hline
				\rowcolor{red!40}
				apr		 &  388	&771&	1\\
				\hline
			\end{tabular}	
		}{}
	\end{floatrow}
\end{figure}

The table above (left) shows the results of running the minimisation algorithm using approximate partition refinement on the LMC that models the bounded retransmission protocol with $N =16$ and $\mathit{MAX} = 3$. The perturbed LMCs are obtained by perturbing the probabilities with $\epsilon = 0.01$. The minimisation algorithm fails to recover the structure of the quotient of the original model in this case. This is due to the fact that the only $\epsilon_2$ in this experiment that is greater than $\epsilon$ is $0.1$. This value, however, is too big and the minimisation algorithm aggressively merges states in the perturbed model.  The table above (right) shows the results of running the minimisation algorithm using approximate partition refinement on the LMC that models the bounded retransmission protocol (with $N =64$ and $\mathit{MAX} = 4$). The perturbed LMCs are obtained by perturbing the probabilities with $\epsilon = 0.001$. The minimisation algorithms successfully recover the structure of the quotient of the original model when $\epsilon_2 = 0.01$. When $\epsilon_2 = 0.1$, it is again too big and the minimisation algorithm aggressively merges states in the perturbed model.


The table above (left) shows the results of running the minimisation algorithms using approximate partition refinement on the LMC that models the Crowds protocol \cite{ReiterR98} with $\mathit{TotalRuns}=4$ and $\mathit{CrowdSize}=5$. The perturbed LMCs are obtained by perturbing the probabilities with $\epsilon = 0.001$. The table above (left) shows the results of running the minimisation algorithms using approximate partition refinement on the LMC that models the Crowds protocol \cite{ReiterR98} with $\mathit{TotalRuns}=6$ and $\mathit{CrowdSize}=5$. The perturbed LMCs are obtained by perturbing the probabilities with $\epsilon = 0.01$.  For both models, the minimisation algorithms successfully recover the structure of the quotient of the original model when $\epsilon_2 \gr \epsilon$. 

\begin{figure}[t!]
	\begin{floatrow}
		\capbtabbox[.49\textwidth]{
		\noindent\begin{tabular}{|c|c|c|c|}
			\hline 
			\multicolumn{4}{|c|}{Crowds4-5}\\
			\hline
			\multirow{1}{*}{\shortstack[l]{$\epsilon = 0.0001$}}& 			
			\# states&	\# trans&	\# iter\\
			\hline 						
			$\M$ \& $\Hyp$	&		 	3515 &	6035				&\\
			$\M/_{\sim_{\M}}$	&		 		34	  &			42				&\\		
			$\Hyp/_{\sim_{\Hyp}}$      & 		1679 &	4199            & \\
			\hline
			\multicolumn{4}{|c|}{Perturbed LMC \#5}\\ 
			\hline
			\multicolumn{4}{|c|}{$\epsilon_2 = 0.00001$}\\
			\hline
			apr      & 1656	& 4087 &	1  \\
			\hline
			\multicolumn{4}{|c|}{$\epsilon_2 = 0.0001$}\\
			\hline
			apr		  &735	&1320&	4 \\
			\hline					
			\multicolumn{4}{|c|}{$\epsilon_2 = \{0.001, 0.01, 0.1\}$}\\
			\hline
			\rowcolor{yellow}
			apr		   & 34	&42	&1   \\
			\rowcolor{white}
			\hline
			\multicolumn{4}{c}{}\\
			\multicolumn{4}{c}{}\\
		\end{tabular}
	}{}
	\capbtabbox[.49\textwidth]{
		\begin{tabular}{|c|c|c|c|}
			\hline
			\multicolumn{4}{|c|}{Crowds6-5}\\
			\hline
			\multirow{1}{*}{\shortstack[l]{$\epsilon = 0.001$}}& 			
			\# states&	\# trans&	\# iter\\
			\hline 
			$\M$ \& $\Hyp$ &	18817	& 32677    & \\
			$\M/_{\sim_{\M}}$	&		 		50	&  62	  		&\\								
			$\Hyp/_{\sim_{\Hyp}}$     & 	9237 &	23097      & \\
			\hline
			\multicolumn{4}{|c|}{Perturbed LMC \#2}\\ 
			\hline
			\multicolumn{4}{|c|}{$\epsilon_2 = 0.00001$}\\
			\hline
			apr		&  9235	& 23087 &	1\\
			\hline
			\multicolumn{4}{|c|}{$\epsilon_2 = 0.0001$}\\
			\hline
			apr		&  9086 &	22358	& 2\\
			\hline			
			\multicolumn{4}{|c|}{$\epsilon_2 =  0.001$}\\
			\hline
			apr		 &  5249 &	9746 &	5 \\
			\hline			
			\multicolumn{4}{|c|}{$\epsilon_2 \in \{0.01,  0.1\}$}\\
			\hline
			\rowcolor{yellow}
			apr		 &  50	&62&	1\\
			\hline
		\end{tabular}	
	}{}
\end{floatrow}
\end{figure}

\begin{figure}[h!]
	\begin{floatrow}
		\capbtabbox[.49\textwidth]{
			\noindent\begin{tabular}{|c|c|c|c|}
				\hline 
				\multicolumn{4}{|c|}{EGL5-2}\\
				\hline
				\multirow{1}{*}{\shortstack[l]{$\epsilon = 0.0001$}}& 				
				\# states&	\# trans&	\# iter\\
				\hline 						
				$\M$ \& $\Hyp$	&		 	33790 &	34813				&\\
				$\M/_{\sim_{\M}}$	&		 		472	 & 507			&\\		
				$\Hyp/_{\sim_{\Hyp}}$      & 		551	& 681            & \\
				\hline
				\multicolumn{4}{|c|}{Perturbed LMC \#3}\\ 
				\hline
				\multicolumn{4}{|c|}{$\epsilon_2 = 0.00001$}\\
				\hline
				apr      & 535	& 647 & 	1  \\
				\hline
				\multicolumn{4}{|c|}{$\epsilon_2 = 0.0001$}\\
				\hline
				apr		  &485 &	536&	2 \\
				\hline					
				\multicolumn{4}{|c|}{$\epsilon_2 = \{0.001, 0.01, 0.1\}$}\\
				\hline
				\rowcolor{yellow}
				apr		   & 472 &	507	& 1   \\
				\rowcolor{white}
				\hline
				\multicolumn{4}{c}{}\\
				\multicolumn{4}{c}{}\\
			\end{tabular}
		}{}
		\capbtabbox[.49\textwidth]{
			\begin{tabular}{|c|c|c|c|}
				\hline
				\multicolumn{4}{|c|}{EGL5-4}\\
				\hline
				\multirow{1}{*}{\shortstack[l]{$\epsilon = 0.001$}}& 						
				\# states&	\# trans&	\# iter\\
				\hline 
				$\M$ \& $\Hyp$ &	74750 &	75773    & \\
				$\M/_{\sim_{\M}}$	&		 		992	& 1027	  		&\\								
				$\Hyp/_{\sim_{\Hyp}}$     & 	1071 &	1201     & \\
				\hline
				\multicolumn{4}{|c|}{Perturbed LMC \#4}\\ 
				\hline
				\multicolumn{4}{|c|}{$\epsilon_2 = 0.00001$}\\
				\hline
				apr		&  1063	& 1184 &	1\\
				\hline
				\multicolumn{4}{|c|}{$\epsilon_2 = 0.0001$}\\
				\hline
				apr		&  1053	& 1162 &	1\\
				\hline			
				\multicolumn{4}{|c|}{$\epsilon_2 =  0.001$}\\
				\hline
				\rowcolor{yellow}
				apr		 &  992	& 1027	& 2 \\
				\rowcolor{white}
				\hline			
				\multicolumn{4}{|c|}{$\epsilon_2 \in \{0.01,  0.1\}$}\\
				\hline
				\rowcolor{yellow}
				apr		 &  992	& 1027 &	1\\
				\hline
			\end{tabular}	
		}{}
	\end{floatrow}
\end{figure}

The table above (left) shows the results of running the minimisation algorithms using approximate partition refinement on the LMC that models the contract signing protocol by Even, Goldreich and Lempel \cite{EvenGL85} with $N=5$ and $L=2$. The perturbed LMCs are obtained by perturbing the probabilities with $\epsilon = 0.0001$. The table above (right) shows the results of running the minimisation algorithm using approximate partition refinement on the LMC that models the contract signing protocol by Even, Goldreich and Lempel \cite{EvenGL85} with $N=5$ and $L=4$. The perturbed LMCs are obtained by perturbing the probabilities with $\epsilon = 0.001$. For both models, the minimisation algorithms successfully recover the quotient of the original model when $\epsilon_2 \gr \epsilon$ ($\epsilon_2 \ge \epsilon$ for the second model).

\end{document}